\documentclass[journal]{IEEEtran}
\usepackage[final]{graphicx}
\usepackage{subfigure}
\usepackage{float}
\usepackage{amsmath}
\usepackage{cases}
\usepackage{color}
\usepackage{colortbl}
\usepackage{algorithm}
\usepackage{algpseudocode}
\usepackage{amsmath}
\usepackage{amssymb}
\usepackage{booktabs}
\usepackage{setspace}
\usepackage{array}
\usepackage{rotating}
\usepackage[normalem]{ulem}
\usepackage{bbding}

\renewcommand{\normalsize}{\fontsize{10}{12}\selectfont}

\newtheorem{theorem}{ \bf Theorem}
\newenvironment{proof}{{ \noindent \it Proof.}}{\hfill $\blacksquare$}

\makeatletter
\renewcommand{\maketag@@@}[1]{\hbox{\m@th\normalsize\normalfont#1}}%
\makeatother
\usepackage{stfloats}
\usepackage{cite}
\usepackage{makecell}
\usepackage{multirow}
\usepackage{hyperref}

\ifCLASSINFOpdf
\else
\fi

\usepackage{caption}
\usepackage{mathtools}

\begin{document}
\title{Multipath Component-Enhanced Signal Processing for Integrated Sensing and Communication Systems}
\author{Haotian Liu,
Zhiqing Wei, 
Xiyang Wang,
Huici Wu,
Fan Liu,
Xingwang Li,
Zhiyong Feng

\thanks{H. Liu, Z. Wei, X. Wang, H. Wu, Z. Feng are with the Beijing University of Posts and Telecommunications, Beijing 100876, China (emails: \{haotian\_liu; weizhiqing; wangxiyang; dailywu; fengzy\}@bupt.edu.cn). \textit{Corresponding authors: Zhiqing Wei, Zhiyong Feng.}
}
\thanks{X. Li is with the School of Physics and Electronic Information Engineering, Henan Polytechnic University, Jiaozuo 454000, China (e-mail: lixingwangbupt@gmail.com).}

\thanks{F. Liu is with School of System Design and Intelligent Manufacturing, Southern University
of Science and Technology, Shenzhen 518055, China (e-mail: f.liu@ieee.org).}
}

\maketitle

\begin{abstract} 
Integrated sensing and communication (ISAC) has gained traction in academia and industry.
Recently, multipath components (MPCs), as a type of spatial resource,
have the potential to improve the sensing performance in ISAC systems,
especially in richly scattering environments.
In this paper, we propose to leverage MPC and Khatri-Rao space-time (KRST) code within a single ISAC system
to realize high-accuracy sensing for multiple dynamic targets and multi-user communication.
Specifically, we propose a novel MPC-enhanced sensing processing scheme with symbol-level fusion,
referred to as the ``SL-MPS'' scheme,
to achieve high-accuracy localization of multiple dynamic targets and
empower the single ISAC system with a new capability of
absolute velocity estimation for multiple targets
with a single sensing attempt.
Furthermore, the KRST code is applied to flexibly balance
communication and sensing performance in richly scattering environments.
To evaluate the contribution of MPCs,
the closed-form Cramér-Rao lower bounds (CRLBs) of location and absolute velocity estimation are derived.
Simulation results illustrate that the proposed SL-MPS scheme
is more robust and accurate in localization and absolute velocity estimation
compared with the existing state-of-the-art schemes.
\end{abstract}
\begin{IEEEkeywords}
Cramér-Rao lower bound (CRLB),
integrated sensing and communication (ISAC),
multi-resource cooperative ISAC,
multipath component (MPC), 
tensor decomposition.
\end{IEEEkeywords}

\IEEEpeerreviewmaketitle

\section{Introduction}

\subsection{Background and Motivations}
\IEEEPARstart{T}{he} envisioned next-generation wireless networks aim to endow mobile communication systems with
integrated sensing and communication (ISAC) capability,
promoting the deep development of emerging application scenarios 
such as the intelligent transportation and low-altitude economy~\cite{Fanliu_2022,wei2024deep,zhang2025,liu2023isac}.
To further enhance the sensing performance of ISAC systems, 
multi-resource cooperative sensing has been widely applied 
in the emerging scenarios~\cite{wei2024deep,liu2024carrier}. 
However, dense multipath components (MPCs) caused by buildings or traffic clusters 
in urban environments severely interfere with line-of-sight (LoS) echo signals,
significantly degrading sensing accuracy. 
Therefore, mitigating the impact of MPCs on high-precision sensing is a critical challenge for ISAC 
systems~\cite{cai2023,laoudias2018survey,gennarelli2014multipath}.

Currently, many approaches focus on detecting and eliminating the effects of MPCs,
such as the maximum likelihood estimation (MLE)~\cite{xie2017mitigating} 
and the CLEAN method~\cite{cramer2002evaluation}.
On top of that, existing works also aim to mitigate the impact of MPC signals on positioning,
utilizing techniques like machine learning-based MPC mitigation~\cite{pan2022machine} and
differential global navigation satellite system-based MPC error elimination~\cite{lee2023seamless}.
While the above methods are sophisticatedly designed, they fall short in maintaining accuracy when there is no LoS path.
To further enhance the sensing performance, one may follow the philosophy of ``turning enemies into allies'' through exploiting the MPCs as useful spatial-domain resources~\cite{wei2024deep}, rather than completely eliminating them.


\subsection{Related Work}
The researches on MPC utilization mainly focus on MPC exploitation radar (MER) and wireless localization, which we review in detail as below.

\textbf{\textit{MER:}} 
The principle of MER is to use mirror-reflected MPC signals
generated by electromagnetic waves on building surfaces,
with prior knowledge of the building layout,
to achieve non-line-of-sight (NLoS) positioning of targets~\cite{lingjiang2023overview,zetik2012uwb,johansson2016positioning,chen2022non,wei2021nonline}. 
Zetik \textit{et al.} in \cite{zetik2012uwb} used single-reflection echoes and
geometric relationships from walls to estimate the one-dimensional coordinates of NLoS targets.
Johansson \textit{et al.} in \cite{johansson2016positioning}
extended the localization of NLoS targets to two-dimensional (2D) cases,
employed narrow-beam scanning radar and utilized multiple mirror reflections for 2D localization of moving targets.
In addition to utilizing MPC's distance information, 
combining angle-of-arrival (AoA) information can improve localization accuracy. 
Chen \textit{et al.} in \cite{chen2022non} introduced an MPC recognition method 
for matching MPCs and multiple targets, 
achieving high-accuracy NLoS multi-target localization.
Beyond estimating 2D coordinates, 
acquiring three-dimensional (3D) height information of targets is also crucial. 
Based on the 3D imaging model, 
Wei \textit{et al.} in \cite{wei2021nonline} reconstructed a rough 3D image from 2D radar images, 
and proposed an automatic focusing algorithm 
for high-accuracy 3D image reconstruction. 

\textbf{\textit{Wireless localization:}} 
As a high-precision positioning technology, 
wireless localization also considers 
``turning enemies into allies'' with respect to MPCs~\cite{witrisal2016high,shen2009use}. 
Ultra-wideband (UWB) technologies, thanks to their easy separability in mixed MPCs, 
are recognized as one of the promising approachs for MPC-assisted wireless
localization~\cite{witrisal2016high,shen2009use,meissner2010uwb,meissner2013accurate}. 
Based on UWB signals, Shen \textit{et al.} in \cite{shen2009use} 
demonstrated the contribution of MPCs on wireless localization through the derivation of the squared position error bound. 
Moreover, Meissner \textit{et al.} in \cite{meissner2010uwb} introduced a method for target localization
using the distance information from MPCs. 
By combining the target distance information carried in the channel impulse response with MLE,
target position information is obtained. 
Based on the previous work in \cite{meissner2010uwb}, 
Meissner \textit{et al.} in \cite{meissner2013accurate} extended
the single-node indoor positioning algorithm to a multi-node collaborative scenario. 
In addition to UWB signals, 
multiple input multiple output-orthogonal frequency division multiplexing (MIMO-OFDM) signals 
are also applied to distinguish and utilize MPCs through tensor decomposition.
Zhao \textit{et al.} in \cite{zhao2021multipath} proposed an MPC separation method based on tensor decomposition, 
using the time of arrival and AoA information of LoS MPCs for single-target localization. 
However, \cite{zhao2021multipath} did not utilize the sensing information of NLoS MPCs.
To address this dilemma, Gong \textit{et al.} in \cite{gong2022multipath} proposed an MPCs-assisted localization method. 
By establishing virtual anchors (VAs), 
joint utilization of both LoS and NLoS MPC information is enabled for localization. 
This method fuses the target parameters carried by MPCs, which may be treated as the data-level fusion of sensing information.

Overall, the above works provide guidance and feasibility analysis 
for the sensing signal processing in MPCs-aided ISAC systems. 
However, the multipath utilization schemes based on MER and UWB 
have the following shortcomings when applied to mobile communication systems.
\begin{itemize}
    \item \textbf{Dense MPC separation:} 
    Dense MPCs increase multipath separation time,
    hindering existing schemes from meeting the real-time demands
    of mobile communication systems.
    \item \textbf{Exploitation of the spatial-temporal-frequency domains}:
    The spatial-temporal-frequency domains capture higher-order MPC features
    and multidimensional gains unique to MIMO-OFDM signals.
    However, existing schemes failed to fully exploit this advantage.
\end{itemize}
It is worth mentioning that the method in~\cite{gong2022multipath}
represents the state-of-the-art
in MPC-assisted research under MIMO-OFDM signals.
However, this method is not suitable for the scenario with multiple dynamic targets.
Furthermore, these techniques typically adopted the conventional data-level fusion principle,
which relies on the geometric relationship between MPC information and targets for localization,
without leveraging the complex spatial characteristics of
MIMO-OFDM signals to further exploit the value of MPC information.

\subsection{Our Contributions}
This paper proposes a novel MPC-enhanced sensing processing scheme with symbol-level fusion,
referred to as the ``SL-MPS'' scheme.
On this basis, a Khatri-Rao spacetime (KRST) code, leveraging the sparsity of its coding matrix and its ability to resist multipath interference, is applied to flexibly balance communication and sensing performance in a richly scattering environment. 
Compared to existing data-level fusion schemes in~\cite{zhao2021multipath} -\cite{gong2022multipath}, 
the proposed SL-MPS scheme offers enhanced sensing performance, 
as symbol-level fusion capitalizes on the intricate distribution patterns of 
MPC information within the complex space.
The main contributions of this paper are summarized as follows.
\begin{itemize}
    \item \textbf{MPC-enabled sensing enhancement:} 
    We consider the signal processing of leveraging MPC information and KRST code to enhance sensing performance in multiple dynamic target and multi-user scenario. 
    \item \textbf{Dense MPC separation algorithm:}
    Based on the sparse redundancy introduced by KRST coding, we propose a transform-predict alternating least squares (TP-ALS) algorithm, which initializes the solution space using higher-order singular value decomposition (HOSVD), imposes constraints on matrix factors through physical model mapping, and employs adaptive moment estimation to dynamically adjust the iteration steps, achieving accurate and fast separation of dense MPCs. 
    \item \textbf{A novel MPC-enabled sensing processing scheme:} 
    An MPC-enhanced sensing processing scheme with symbol-level fusion is proposed, 
    referred to as the SL-MPS scheme.
    This scheme turns harmful MPCs into valuable sensing information and 
    combines the developed symbol-level fusion operation to achieve 
    enhanced single-station localization performance while empowering single ISAC system with the capability of
    absolute velocity estimation for multiple targets with a single sensing attempt.
    \item \textbf{CRLB analysis of MPC-enabled sensing:} 
    The CRLBs of localization and absolute velocity estimation with MPC information are derived, 
    revealing the contribution of MPC information to moving target localization and absolute velocity estimation.
\end{itemize}

This paper significantly extends the method proposed in our conference paper~\cite{liu2024multipath}.
We further consider the practical scenarios with multiple dynamic targets
and multiple user equipments (UEs).
This introduces the following new challenges: 
\begin{enumerate}
    \item The accuracy and speed of MPC separation is difficult to be improved due to dense MPCs.
    \item The association of targets, reflectors and MPCs is difficult in this scenario.
    \item The application of Doppler information from MPCs in sensing dynamic targets is challenging.
\end{enumerate}
To this end, we propose a TP-ALS algorithm and matching operation for multi-dynamic target separation and association.
Notably, the single ISAC system is empowered with a new capability of absolute velocity estimation 
for multiple targets with a single sensing attempt.
Additionally, we provide the detailed closed-form CRLBs 
of MPC-enabled location and absolute velocity estimations.
The proposed schemes are evaluated by simulation results,
including the trade-off analysis of the dual-function performance 
and the performance evaluation of the proposed sensing processing scheme.

The remainder of this paper is structured as follows. 
Section~\ref{se2} outlines the system model. 
In Section~\ref{se3}, we present the SL-MPS scheme. 
Section~\ref{se4} derives the CRLBs of localization and absolute velocity estimation with MPC information. 
Section \ref{se5} details the simulation results, and Section \ref{se6} summarizes this paper.

\textit{Notations:} 
$\{\cdot\}$ typically stands for a set of various index values. 
The bold black letters represent matrices or vectors. 
$\left \langle \cdot \right \rangle $ denotes extracting data from a cell $(\alpha, \beta)$.
$\left[\!\left[\cdot \right]\!\right]$ is the Kruskal operator.
$\mathbb{C}$ and $\mathbb{R}$ denote the set of complex and real numbers, respectively. 
$\left[\cdot\right]^{\text{T}}$, 
$\left[\cdot\right]^{\text{H}}$, $\left[\cdot\right]^{-1}$, 
$\left[\cdot\right]^{\dagger}$,
$|\cdot|$, $\text{diag}\{\cdot\}$, 
$\mathrm{real}(\cdot)$, and $\propto $ are the transpose, conjugate transpose, inverse, pseudo-inverse, absolute, diagonal, real part, and proportionality operators, respectively. 
$\text{vec}(\cdot)$ is the vector operator. 
$\equiv$ and $\approx $ are identity operator and approximation operator, respectively. 
$\bullet$, $\circ $, and $\odot$ denote the inner, outer, and Khatri–Rao products, respectively. 
$\left\|\cdot\right\|_{\text{F}}^2$ and $\|\cdot\|_\infty$ are the Frobenius and infinity norm, respectively. 
A complex Gaussian random variable $\mathbf{u}$ with mean $\mu_u$ and variance $\sigma_u^2$ is denoted by $\mathbf{u} \sim \mathcal{CN}\left(\mu_u,\sigma_u^2\right)$.

\section{System Model}\label{se2}
As shown in Fig. \ref{fig1}, 
we consider an ISAC-enabled base station (BS) for multiple dynamic target sensing and multi-user downlink (DL) communication in a richly scattering environment. 
The ISAC BS is deployed with uniform linear arrays with the antenna spacing being $d_r$, 
where $N_\text{T}$ transmit antennas and $N_\text{R}$ receive antennas are time-shared for communication and sensing~\cite{liu2024carrier}.
For multi-user DL communication, 
we assume that there are $U$ UEs, 
each equipped with $N_\text{U}$ receive antennas. 
For multiple dynamic target sensing,
we assume that there are $I$ targets and that NLoS MPCs and LoS MPCs exist between the BS and targets. 
It is noted that our work is also applicable when the LoS MPCs do not exist, where the NLoS MPC refers to a symmetric two-way path in which the transmitted signal reaches the target via a single reflection and returns along the same path in reverse.
Given that the BS and reflectors are stationary, we assume the BS has prior knowledge of the positions and angles of the $L$ reflectors, 
whose surfaces yield only specular reflections of the ISAC signal~\cite{sen2010adaptive}. 
Due to the extremely low energy of multiple reflections, 
it is common in the literature to consider only a single reflection~\cite{sen2010adaptive,gong2022multipath,witrisal2016high}. 
Besides, we also consider the presence of clutter, such as MPC where the signal is reflected through an unknown moving scatterer, which is called temporary MPC in the paper.
The total number of various MPCs received by BS can be determined using the minimum description length (MDL) method~\cite{yokota2016robust}, 
which is assumed to be $I\left(L+1+K_s\right)$, 
with each target corresponding to one LoS MPC, $L$ NLoS MPCs, and $K_s$ temporary MPCs.
\begin{figure}
    \centering
\includegraphics[width=0.48\textwidth]{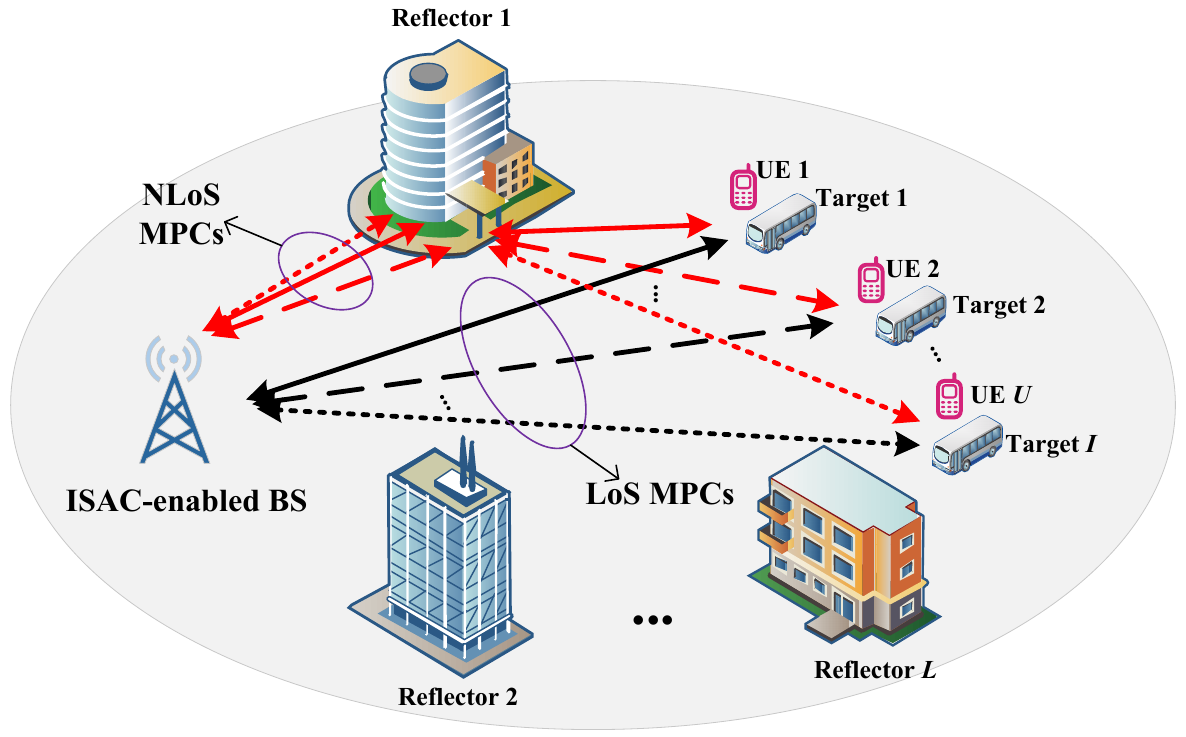}
    \caption{MPCs-enhanced MIMO-OFDM ISAC BS for multi-user communication and multiple dynamic target sensing scenario}
    \label{fig1}
\end{figure}

\subsection{KRST-MIMO-OFDM ISAC Transmitter}
In a richly scattering environment, KRST leverages the Khatri-Rao product structure to reduce the correlation of multipath components (MPCs) for easier separation while suppressing multipath fading through diversity combining. Therefore, we consider the DL communication and sensing with $M$ KRST code blocks.
For the $u$-th UE, 
the modulated symbol matrix associated with the $n$-th subcarrier is $\mathbf{X}_n^u=\left[\mathbf{x}_n^{1,u},\cdots,\mathbf{x}_n^{m,u},\cdots,\mathbf{x}_n^{M,u}\right]\in\mathbb{C}^{N_\text{S}\times M}$, 
where $N_\text{S}$ is the number of parallel modulated symbol streams, 
and $n\in\{0,1,\cdots,N_\text{c}\}$ is the index of subcarrier with the total number being $N_\text{c}$. 
$\mathbf{X}_n^u$ is first precoded by $\mathbf{W}_n^u\in\mathbb{C}^{N_\text{T}\times N_\text{S}}$~\cite{Lu2024random}, 
and the overlapped precoded modulated symbol matrix is
\begin{equation}\label{eq1}
\mathbf{S}_n= \sum_{u=1}^{U}\mathbf{W}_n^u\mathbf{X}_n^u 
=\left[\mathbf{s}_n^1,\cdots,\mathbf{s}_n^m,\cdots,\mathbf{s}_n^M\right]\in\mathbb{C}^{N_\text{T}\times M}.
\end{equation}

Then, for the $m$-th code block, 
the coding symbol matrix encoded by KRST coding scheme~\cite{sidiropoulos2002khatri} is $\text{diag}\left(\Xi\mathbf{s}_n^m\right)\mathbf{C}_0^\text{T}$, 
where $\Xi=\frac{1}{\sqrt{N_\text{T}}}\tilde{\boldsymbol{F}}_{n}\text{diag}\left\{1,e^{j\pi/2N_\text{T}},\cdots,e^{j\pi(N_\text{T}-1)/2N_\text{T}}\right\}\in\mathbb{C}^{N_\text{T} \times N_\text{T}}$ 
is a constellation rotation matrix with $\tilde{\boldsymbol{F}}_{n}\in\mathbb{C}^{N_\text{T} \times N_\text{T}}$ being an inverse discrete Fourier transform (IDFT) matrix; 
$\mathbf{C}_0\in\mathbb{C}^{K\times N_\text{T}}$ is a scaled semi-unitary matrix with $\mathbf{C}_0^\text{H}\mathbf{C}_0=\mathbf{I}_{N_\text{T}}\left(K\le N_\text{T}\right)$ 
and $K$ being the code length (number of slots in one code).

\subsection{Communication Signal Model}
In the richly scattering environment,
the KRST code is introduced to combat multipath fading while enabling flexible adjustment of full rate or 
full diversity gain~\cite{sidiropoulos2002khatri}.

\subsubsection{Channel model}
Since a LoS MPC exists between BS and the $u$-th UE, 
we assume that the channel matrix $\mathbf{H}_C^u\in\mathbb{C}^{N_\text{U}\times N_\text{T}}$ associated with the $u$-th UE 
is an uncorrelated Rician fading channel and its elements can be expressed as~\cite{sanguinetti2018theoretical}
\begin{equation}\label{eq2} \left[\mathbf{H}_C^u\right]_{\omega,p}=\sqrt{\beta\frac{\kappa}{1+\kappa}}h_{\omega,p}^{Ri}+\sqrt{\beta\frac{1}{1+\kappa}}h_{\omega,p}^{Ra},
\end{equation}
where $\omega\in\{0,1,\cdots,N_\text{U}-1\}$ and $p\in\{0,1,\cdots,N_\text{T}-1\}$ 
are the indices of the receive UE and transmit BS antennas, respectively; 
$\beta = \frac{(4\pi)^2 d^2}{\lambda^2}$ is pathloss with $\lambda$ being wavelength 
and $d$ being the distance from BS to UE;
$\kappa \ge 0$ is the Rician factor; 
$h_{\omega,p}^{Ri}$ is a deterministic
vector that associates with the LoS path and $h_{\omega,p}^{Ra}\sim\mathcal{CN}\left(0,1
\right)$~\cite{sanguinetti2018theoretical}. 

\subsubsection{Received communication signal}
If the channel is constant for the duration of code length $K$, 
the received communication signal at 
the $u$-th UE in the $n$-th subcarrier during the $m$-th KRST code block is denoted by
\begin{equation} \label{eq3}
\mathbf{Y}_{n,m}^{C,u}=\mathbf{H}_C^u\text{diag}\left(\Xi\mathbf{s}_n^m\right)\mathbf{C}_0^\text{T}+\mathbf{Z}_{n,m}^{C,u},
\end{equation}
where $\mathbf{Y}_{n,m}^{C,u}\in\mathbb{C}^{N_\text{U}\times K}$ and 
$\mathbf{Z}_{n,m}^{C,u}\in\mathbb{C}^{N_\text{U}\times K}$ is additive white Gaussian noise (AWGN) matrix. 
The KRST code enables communication data recovery without requiring channel state information, 
which is realized by the blind-KRST decoding technique~\cite{sidiropoulos2002khatri}, 
and the received modulated symbol is denoted by $\tilde{\mathbf{s}}_n^m$. 
With the known $\mathbf{W}_n^u$ and $\tilde{\mathbf{s}}_n^m$, 
the received modulated symbol of the $u$-th UE is $\hat{\mathbf{X}}_n^u$. 
With the KRST code, 
the transmission rate in one subcarrier is $\left(\frac{N_\text{T}}{K}\right)\log_2\left(\varphi \right)$ bit/code length, 
where $\varphi$ is the modulation order.

\subsection{Sensing Signal Model}
For the sensing in the richly scattering environment, 
based on the sparsity of KRST codes, 
we aim to turn MPC information into valuable resources 
to enhance sensing performance in multiple dynamic target sensing.

In the $M$ KRST code blocks sensing duration, 
for $I$ far-field point targets moving with the $i\in\{1,2,\cdots,I\}$-th target's absolute velocity 
being $\vec{v}_i$ and coordinate being $\left(x_\text{ta}^i,y_\text{ta}^i\right)$, 
the echo signal in BS side of the $n$-th subcarrier during the $m$-th KRST code block is~\cite{liu2024carrier}
\begin{equation}  \label{eq4}
\mathbf{Y}_{n,m}^S=\sum_{i=1}^I\sum_{l=0}^{L+K_s}\left[
\begin{array}{l}
 \sqrt{P_{\text{t},i}^l}\alpha_i^l e^{j2\pi f_{\text{D},i}^l mT}e^{-j2\pi n\Delta f \tau_i^l}  \\
\times\textbf{a}_\text{r}\left(\theta_i^l\right)\boldsymbol{\chi}\text{diag}\left(\Xi\mathbf{s}_n^m\right)\mathbf{C}_0^\text{T}
\end{array}\right]+ \mathbf{Z}_{n,m}^S,
\end{equation}
where $\mathbf{Y}_{n,m}^S\in \mathbb{C}^{N_\text{R}\times K}$, 
and $l\in\{0,1,\cdots, L+K_a\}$ denote the index of MPCs; 
$l=0$ is the LoS MPC for the $i$-th target; 
$\alpha_i^l$ and $P_{\text{t},i}^l$ are the attenuation and equivalent transmit power of the $l$-th MPC, respectively; 
$\sum_{i=1}^I\sum_{l=0}^{L}P_{\text{t},i}^l$ is the total transmit power;
$f_{\text{D},i}^l=\frac{2v_i^lf_\text{c}}{c_0}$ is the Doppler frequency shift, 
where $v_i^l$ is the relative velocity with respect to the $l$-th MPC for the $i$-th target and $c_0$ is the speed of light;
$\tau_i^l=\frac{2R_i^l}{c_0}$ is the delay with $R_i^l$ being the one-way distance of the $l$-th MPC for the $i$-th target; 
$T$ is total OFDM symbol duration;
$\theta_i^l$ is the AoA of the $l$-th MPC for the $i$-th target arrives at BS and
$\mathbf{Z}_{n,m}^S \sim \mathcal{CN}(0,\sigma^2)$ is an AWGN matrix; 
$\boldsymbol{\chi}_i^l\in\mathbb{R}^{1\times N_\text{T}}$ is the transmit beamforming gain vector with $\boldsymbol{\chi}_i^l(p) = \textbf{a}_\text{t}^\text{T}\left(\theta_i^l\right)\mathbf{w}_{\text{t},i}^l$ and $\mathbf{w}_{\text{t},i}^l\in\mathbb{C}^{N_\text{T}\times 1}$ being transmit beamforming vector obtained based on beam scanning; 
$\textbf{a}_\text{r}(\cdot)$ and $\textbf{a}_\text{t}(\cdot)$ are the receive and transmit steering vectors, 
expressed in (\ref{eq5}) and (\ref{eq6}), respectively~\cite{liu2024carrier}.
\begin{equation}\label{eq5}
    \textbf{a}_\text{r}(\cdot)=\left[1, \cdots, e^{j2\pi\frac{d_r}{\lambda} k\sin(\cdot)}, \cdots, e^{j2\pi \frac{d_r}{\lambda}(N_\text{R}-1)\sin(\cdot)}\right]^\text{T},
\end{equation}
\begin{equation}\label{eq6}
\textbf{a}_\text{t}(\cdot)=\left[1, \cdots, e^{j2\pi\frac{d_r}{\lambda} p\sin(\cdot)}, \cdots, e^{j2\pi\frac{d_r}{\lambda} (N_\text{T}-1)\sin(\cdot)}\right]^\text{T},
\end{equation}
where $k\in\{0, 1, \cdots, N_\text{R}-1\}$ is the index of receive BS antennas.  
According to Eq.~\eqref{eq4}, $I(L+1+K_s)$ MPCs carrying target information show promise for sensing, yet the instability of temporary MPCs leads us to presently discount their use.

\section{MPC-Enhanced Sensing Signal Processing with Symbol-Level Fusion}\label{se3}
As shown in Fig.~\ref{fig2}, 
the proposed SL-MPS scheme, involves two stages: 
MPC separation stage and symbol-level fusion stage. 
In the MPC separation stage, the separated problem is solved by a proposed TP-ALS algorithm.
In the symbol-level fusion stage, 
the VAs are constructed based on the prior environment information, 
and the one-to-one relationship between MPCs, VAs, and targets 
is obtained by the proposed matching operation, 
which ensures the coupling of sensing information, 
while achieving the decoupling of the multi-target sensing problem. 
Then, the separated sensing information is fused using symbol-level fusion
to estimate the location and absolute velocity of $I$ targets.
\begin{figure*}
    \centering
    \includegraphics[width=0.80\textwidth]{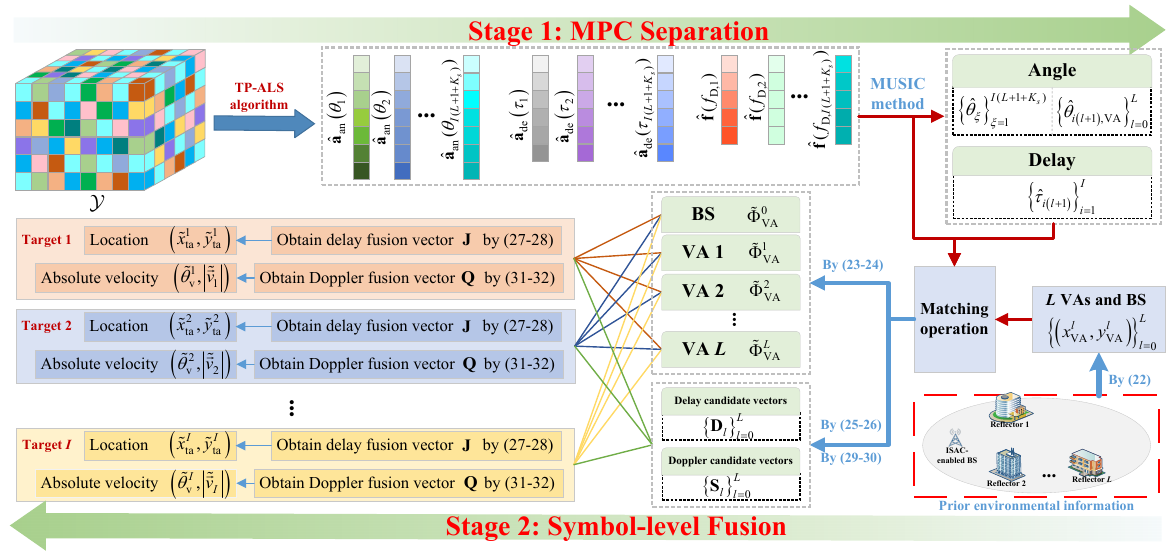}
    \caption{The diagram of the proposed SL-MPS scheme}
    \label{fig2}
\end{figure*}

\subsection{MPC Separation Stage}\label{sec3-a}

Considering the high-dimensional characteristic of the received echo signal 
and the uncorrelated characteristic of the MPC information, 
the tensor decomposition can be applied to separate the mixed MPCs and maintain the coupling of delay, 
angle, and Doppler information of the target~\cite{kolda2009tensor}.
\subsubsection{Tensor-based ISAC echo signal model}
First, we need to cancel the transmission symbol 
to mitigate the influence of its randomness on sensing. 
We multiply $\mathbf{Y}_{n,m}^S$ by $\text{diag}\left[\mathbf{e}\left(\text{diag}\left(\Xi\mathbf{s}_n^m\right)\mathbf{C}_0^\text{T}\right)^*\right]$ 
on the right to get
\begin{equation}
\begin{aligned}
\tilde{\mathbf{Y}}_{n,m}^S& = \mathbf{Y}_{n,m}^S \text{diag}\left[\mathbf{e}\left(\text{diag}\left(\Xi\mathbf{s}_n^m\right)\mathbf{C}_0^\text{T}\right)^*\right] \\ & =\sum_{i=1}^I\sum_{l=0}^{L+K_s}\left[
\sqrt{P_{\text{t},i}^l}\alpha_i^l e^{j2\pi f_{\text{D},i}^l mT}e^{-j2\pi n\Delta f \tau_i^l}\textbf{a}_\text{r}\left(\theta_i^l\right)\tilde{\boldsymbol{\chi}}
\right] \\ & \quad+ \mathbf{Z}_{n,m}^S\text{diag}\left[\mathbf{e}\left(\text{diag}\left(\Xi\mathbf{s}_n^m\right)\mathbf{C}_0^\text{T}\right)^*\right],   
\end{aligned} 
\end{equation}
where $\mathbf{e}=[1_1,1_2,\cdots,1_{N_\text{T}}]\in\mathbb{R}^{1\times N_\text{T}}$, 
and $\tilde{\mathbf{\chi}}\in\mathbb{R}^{1\times K}$ is a gain factor generated by canceling the transmitted symbol.
The phase information of each column of $\tilde{\mathbf{Y}}_{n,m}^S$ is identical, 
and we can coherently accumulate it to obtain a coding gain $K$. 
The accumulated echo signal of the $n$-th subcarrier during the $m$-th KRST code block 
is $\hat{\mathbf{Y}}_{n,m}^S=\frac{1}{K}\sum_{\varepsilon=1}^K\tilde{\mathbf{Y}}_{n,m}^S\left[:,\varepsilon\right]$. 
Then, the third-order tensor form of the accumulated echo signal in $N_\text{c}$ subcarriers 
during $M$ code block can be expressed as~\cite{gong2022multipath}
\begin{equation}\label{eq7}
\begin{aligned}
     \mathcal{Y}&\approx \left[\!\left[ \mathbf{A},\mathbf{B},\mathbf{C} \right]\!\right] + \mathcal{Z}\\& \equiv \sum_{i=1}^{I}\sum_{l=0}^{L+K_s}\mathbf{a}_{\text{an}}\left(\theta_i^l\right)\circ  \mathbf{a}_{\text{de}}\left(\tau_i^l\right)\circ   \mathbf{f}\left(f_{\text{D},i}^l\right)+ \mathcal{Z},
\end{aligned}
\end{equation}
where $\mathcal{Y} \in \mathbb{C}^{N_\text{R}\times N_\text{c}\times M}$, 
and $\mathcal{Z}$ is an AWGN tensor with the element $\mathcal{Z}(k,n,m)$ following $\mathcal{CN}(0, \sigma^2)$; 
The factor matrices $\mathbf{A}=\left[\mathbf{a}_{\text{an}}\left(\theta_1^0\right),\cdots, \mathbf{a}_{\text{an}}\left(\theta_I^{L+K_s}\right)\right]\in \mathbb{C}^{N_\text{R}\times[I(L+1+K_s)]}$, $\mathbf{B}=\left[\mathbf{a}_{\text{de}}\left(\tau_1^0\right),\cdots, \mathbf{a}_{\text{de}}\left(\tau_I^{L+K_s}\right)\right]\in \mathbb{C}^{N_\text{c}\times[I(L+1+K_s)]}$, and $\mathbf{C}=\left[\mathbf{f}\left(f_{\text{D},1}^0\right),\cdots, \mathbf{f}\left(f_{\text{D},I}^{L+K_s}\right)\right]\in \mathbb{C}^{M\times[I(L+1+K_s)]}$; 
$\mathbf{a}_{\text{an}}\left(\theta_i^l\right)$, $\mathbf{a}_{\text{de}}\left(\tau_i^l\right)$, and $\mathbf{f}\left(f_{\text{D},i}^l\right)$ are the vectors of $l$-th MPC for the $i$-th target, expressed in (\ref{eq8}), (\ref{eq9}), and (\ref{eq10}), respectively, in the form of
\begin{equation}\label{eq8}
\mathbf{a}_{\text{an}}\left(\theta_i^l\right)=\textbf{a}_\text{r}\left(\theta_i^l\right)\varpi, 
\end{equation}
\begin{equation}\label{eq9}
\mathbf{a}_{\text{de}}\left(\tau_i^l\right)=\left[1, e^{-j2\pi\Delta f\tau_i^l}, \cdots, e^{-j2\pi(N_\text{c}-1)\Delta f\tau_i^l}\right]^\text{T}, 
\end{equation}
\begin{equation}\label{eq10}
 \mathbf{f}\left(f_{\text{D},i}^l\right)= \left[1, e^{j2\pi f_{\text{D},i}^lT}, \cdots, e^{j2\pi(M-1) f_{\text{D},i}^lT}\right]^\text{T},
\end{equation}
where $\varpi$ is a complex factor.
Then, we present the following remark to prove the uniqueness of decomposition.

\textit{Remark 1: According to Kruskal's sufficient condition for uniqueness, 
$\text{rank}(\mathbf{A})+\text{rank}(\mathbf{B})+\text{rank}(\mathbf{C}) \ge 2I(L+1+K_s)+2$ 
should be satisfied~\cite{kolda2009tensor}. 
Since the angle, delay, and Doppler characteristics of the MPCs for $I$ targets are unique and
$\left\{\mathbf{A},\mathbf{B},\mathbf{C}\right\}$ are full-rank Vandermonde matrices, $\mathcal{Y}$ has a unique tensor decomposition.}

\subsubsection{Proposed {TP-ALS} algorithm}\label{sec3-2}
The mode-1, mode-2, and mode-3 unfoldings $\left[\mathcal{Y}\right]_{(1)}\in\mathbb{C}^{N_\mathrm{R} \times N_\mathrm{c}M}$, $\left[\mathcal{Y}\right]_{(2)}\in\mathbb{C}^{N_\mathrm{c} \times N_\mathrm{R}M}$, and $\left[\mathcal{Y}\right]_{(3)}\in\mathbb{C}^{M  \times N_\mathrm{R}N_\mathrm{c}}$ of $\mathcal{Y}$ can be expressed in \eqref{add1}, \eqref{add2}, and \eqref{add3}, respectively,
where $\mathcal{Y}_{::m+1}\in\mathbb{C}^{N_\mathrm{R} \times N_\mathrm{c}}$ is the $(m+1)$-th frontal slice in the third dimensional, $\mathcal{Y}_{k+1::}\in\mathbb{C}^{N_\mathrm{c} \times M}$ is the $(k+1)$-th horizontal slice in the first dimensional, and $\mathcal{Y}_{:n+1:}\in\mathbb{C}^{N_\mathrm{R} \times M}$ is the $(n+1)$-th lateral slice in the second dimensional.
\begin{equation}\label{add1}
\begin{aligned}
\left[\mathcal{Y}\right]_{(1)}=& \left[\mathcal{Y}_{::1},\mathcal{Y}_{::2},\cdots,\mathcal{Y}_{::M}\right]= \mathbf{A}\left(\mathbf{C}\odot \mathbf{B}\right)^{\text{T}},     
\end{aligned}
\end{equation}
\begin{equation}\label{add2}
\left[\mathcal{Y}\right]_{(2)}=\left[\mathcal{Y}_{1::},\mathcal{Y}_{2::},\cdots,\mathcal{Y}_{N_\mathrm{R}::}\right]=\mathbf{B}\left(\mathbf{C}\odot \mathbf{A}\right)^{\text{T}},     
\end{equation}
\begin{equation}\label{add3}
\begin{aligned}
\left[\mathcal{Y}\right]_{(3)}=& \left[\mathcal{Y}_{:1:}^\mathrm{T},\mathcal{Y}_{:2:}^\mathrm{T},\cdots,\mathcal{Y}_{:N_\mathrm{c}:}^\mathrm{T}\right] =\mathbf{C}\left(\mathbf{B}\odot \mathbf{A}\right)^{\text{T}}.   
\end{aligned}
\end{equation}

In TP-ALS algorithm, the HOSVD leverages SVD decomposition to extract high-order correlations, providing low-rank initial guesses. The physical model mapping narrows the search space through bidirectional transformations between matrix and sparse domains, while adaptive moment estimation balances convergence speed and stability by learning from historical gradients.

\textit{Step 1:} Perform HOSVD on tensor $\mathcal{Y}$ to obtain left singular matrices $\mathbf{U}_1\in\mathbb{C}^{N_\mathrm{R} \times N_\mathrm{R}}$, $\mathbf{U}_2\in\mathbb{C}^{N_\mathrm{c} \times N_\mathrm{c}}$, and $\mathbf{U}_3\in\mathbb{C}^{M \times M}$. Then, the initial factor matrices $\mathbf{\bar{A}}^0 = \mathbf{U}_1[:,1:(I(L+1+K_s))]\in\mathbb{C}^{N_\mathrm{R} \times (I(L+1+K_s))}$, 
$\mathbf{\bar{B}}^0 = \mathbf{U}_2[:,1:(I(L+1+K_s))]\in\mathbb{C}^{N_\mathrm{C} \times (I(L+1+K_s))}$, and
$\mathbf{\bar{C}}^0 = \mathbf{U}_3[:,1:(I(L+1+K_s))]\in\mathbb{C}^{M \times (I(L+1+K_s))}$ are obtained.

\textit{Step 2:} Given the sparsity of the factor matrices in the delay-angle-Doppler domains, we exploit the multiple signal classification (MUSIC) method~\cite{liu2024carrier} to extract the physical parameters carried by each column of $\left\{\mathbf{\bar{A}}_0,\mathbf{\bar{B}}_0,\mathbf{\bar{C}}_0\right\}$, denoted as
$\left\{\hat{\theta}_1^\mathrm{ini},\hat{\theta}_2^\mathrm{ini},\cdots,\hat{\theta}_{I(L+1+K_s)}^\mathrm{ini}\right\}$, $\left\{\hat{\tau}_1^\mathrm{ini},\hat{\tau}_2^\mathrm{ini},\cdots,\hat{\tau}_{I(L+1+K_s)}^\mathrm{ini}\right\}$, and 
$\left\{\hat{f}_{\mathrm{D},1}^\mathrm{ini},\hat{f}_{\mathrm{D},2}^\mathrm{ini},\cdots,\hat{f}_{\mathrm{D},{I(L+1+K_s)}}^\mathrm{ini}\right\}$. Then, referring to the models shown in \eqref{eq8} - \eqref{eq10}, the physical parameters are mapped to the noise-free matrix space, and $\left\{\mathbf{\tilde{A}}_0,\mathbf{\tilde{B}}_0,\mathbf{\tilde{C}}_0\right\}$ are obtained as the final initial factor matrices.

\textit{Step 3:} For the $\vartheta\left(\vartheta\ge1\right)$-th iteration, the analytical solutions of $\left\{\tilde{\mathbf{A}}^{(\vartheta)},\tilde{\mathbf{B}}^{(\vartheta)},\tilde{\mathbf{C}}^{(\vartheta)}\right\}$ obtained by the ALS~\cite{kolda2009tensor} are
\begin{equation}\label{add4}
{\fontsize{11}{11}\begin{aligned}
   &{{\tilde{\mathbf{A}}}^{(\vartheta)}}=\left[\mathcal{Y}\right]_{(1)}
   \left( {{{\tilde{\mathbf{C}}}}^{\left( \vartheta-1 \right)}}\odot {{{\tilde{\mathbf{B}}}}^{\left( \vartheta-1 \right)}} \right)\\ & \times \left[ \left( {{\left( {{{\tilde{\mathbf{C}}}}^{\left( \vartheta -1\right)}} \right)}^\mathrm{H}}{{{\tilde{\mathbf{C}}}}^{\left(\vartheta-1 \right)}} \right)*\left( {{\left( {{{\tilde{\mathbf{B}}}}^{\left( \vartheta-1 \right)}} \right)}^\mathrm{H}}{{{\tilde{\mathbf{B}}}}^{\left( \vartheta-1 \right)}} \right) \right]^{\dagger}, \\
   &{{\tilde{\mathbf{B}}}^{(\vartheta)}}=\left[\mathcal{Y}\right]_{(2)}
   \left( {{{\tilde{\mathbf{C}}}}^{\left( \vartheta-1 \right)}}\odot {{{\tilde{\mathbf{A}}}}^{\left( \vartheta \right)}} \right)\\ & \times \left[ \left( {{\left( {{{\tilde{\mathbf{C}}}}^{\left( \vartheta -1\right)}} \right)}^\mathrm{H}}{{{\tilde{\mathbf{C}}}}^{\left(\vartheta-1 \right)}} \right)*\left( {{\left( {{{\tilde{\mathbf{A}}}}^{\left( \vartheta \right)}} \right)}^\mathrm{H}}{{{\tilde{\mathbf{A}}}}^{\left( \vartheta \right)}} \right) \right]^{\dagger}, \\
   &{{\tilde{\mathbf{C}}}^{(\vartheta)}}=\left[\mathcal{Y}\right]_{(3)}
   \left( {{{\tilde{\mathbf{B}}}}^{\left( \vartheta \right)}}\odot {{{\tilde{\mathbf{A}}}}^{\left( \vartheta \right)}} \right)\\ & \times \left[ \left( {{\left( {{{\tilde{\mathbf{B}}}}^{\left( \vartheta \right)}} \right)}^\mathrm{H}}{{{\tilde{\mathbf{B}}}}^{\left(\vartheta \right)}} \right)*\left( {{\left( {{{\tilde{\mathbf{A}}}}^{\left( \vartheta \right)}} \right)}^\mathrm{H}}{{{\tilde{\mathbf{A}}}}^{\left( \vartheta \right)}} \right) \right]^{\dagger}.
\end{aligned}}
\end{equation}
Given the estimations $\left\{\tilde{\mathbf{A}}^{(\vartheta-1)},\tilde{\mathbf{B}}^{(\vartheta-1)},\tilde{\mathbf{C}}^{(\vartheta-1)}\right\}$ and $\left\{\tilde{\mathbf{A}}^{(\vartheta)},\tilde{\mathbf{B}}^{(\vartheta)},\tilde{\mathbf{C}}^{(\vartheta)}\right\}$, the updated factor matrices under the linear search strategy can be expressed as
\begin{equation}\label{add5}
\begin{aligned}
  \mathbf{\tilde{A}}^{(new)} =  \tilde{\mathbf{A}}^{(\vartheta)} &+ \alpha_{1,\mathrm{opt}}^{(\vartheta)}\Delta_1^{(\vartheta)}, \mathbf{\tilde{B}}^{(new)} =  \tilde{\mathbf{B}}^{(\vartheta)} + \alpha_{2,\mathrm{opt}}^{(\vartheta)}\Delta_2^{(\vartheta)}, \\ &  \mathbf{\tilde{C}}^{(new)} =  \tilde{\mathbf{C}}^{(\vartheta)} + \alpha_{3,\mathrm{opt}}^{(\vartheta)}\Delta_3^{(\vartheta)}.
\end{aligned}
\end{equation}
where $\left\{\alpha_{1,\mathrm{opt}}^{(\vartheta)},\alpha_{2,\mathrm{opt}}^{(\vartheta)},\alpha_{3,\mathrm{opt}}^{(\vartheta)}\right\}$ are the optimal relaxation factors; $\Delta_1^{\vartheta} = \left(\tilde{\mathbf{A}}^{(\vartheta)}-\tilde{\mathbf{A}}^{(\vartheta-1)}\right)$, $\Delta_2^{\vartheta} = \left(\tilde{\mathbf{B}}^{(\vartheta)}-\tilde{\mathbf{B}}^{(\vartheta-1)}\right)$, and $\Delta_3^{\vartheta} = \left(\tilde{\mathbf{C}}^{(\vartheta)}-\tilde{\mathbf{C}}^{(\vartheta-1)}\right)$ denote the updating direction. 

\textit{Step 4:} To obtain $\left\{\alpha_{1,\mathrm{opt}}^{(\vartheta)},\alpha_{2,\mathrm{opt}}^{(\vartheta)},\alpha_{3,\mathrm{opt}}^{(\vartheta)}\right\}$, we predict the initial factors $\left\{\hat{\alpha}_1^{(\vartheta)},\hat{\alpha}_2^{(\vartheta)},\hat{\alpha}_3^{(\vartheta)}\right\}$ by adaptive moment estimation and minimize the following objective function
\begin{equation}\label{add6}
f\left(\alpha_1\right) =\frac{1}{2}\left\|\mathcal{Y}-\left[\!\left[\mathbf{\tilde{A}}^{\vartheta}+\alpha_1\Delta_1^{(\vartheta)}, \mathbf{\tilde{B}}^{\vartheta},\mathbf{\tilde{C}}^{\vartheta}
 \right]\!\right]\right\|_{\mathrm{F}}^2
\end{equation}
\begin{equation}\label{add7}
  f\left(\alpha_2\right) =\frac{1}{2}\left\|\mathcal{Y}-\left[\!\left[\mathbf{\tilde{A}}^{new}, \mathbf{\tilde{B}}^{\vartheta}+\alpha_2\Delta_2^{(\vartheta)},\mathbf{\tilde{C}}^{\vartheta}
 \right]\!\right]\right\|_{\mathrm{F}}^2  
\end{equation}
\begin{equation}\label{add8}
f\left(\alpha_3\right) =\frac{1}{2}\left\|\mathcal{Y}-\left[\!\left[\mathbf{\tilde{A}}^{new}, \mathbf{\tilde{B}}^{new},\mathbf{\tilde{C}}^{\vartheta}+\alpha_3\Delta_3^{(\vartheta)}
 \right]\!\right]\right\|_{\mathrm{F}}^2    
\end{equation}

We take the prediction of $\hat{\alpha}_1^{(\vartheta)}$ as an example, and the $\left\{\hat{\alpha}_2^{(\vartheta)},\hat{\alpha}_3^{(\vartheta)}\right\}$ are similar. Taking the $(\vartheta-1)$-th optimal relaxation factor $\alpha_{1,\mathrm{opt}}^{(\vartheta-1)}$ as the starting point and obtaining the derivative of \eqref{add6}
\begin{equation}\label{add9}
\begin{aligned}
    g_1^{(\vartheta)} = & -\mathrm{real}\left(\left[\!\left[\Delta_1^{(\vartheta)}, \mathbf{\tilde{B}}^{\vartheta},\mathbf{\tilde{C}}^{\vartheta}
 \right]\!\right] \bullet  \right. \\ & \left. \left(\mathcal{Y}-\left[\!\left[\mathbf{\tilde{A}}^{\vartheta}+\alpha_{1,\mathrm{opt}}^{(\vartheta-1)}\Delta_1^{(\vartheta)}, \mathbf{\tilde{B}}^{\vartheta},\mathbf{\tilde{C}}^{\vartheta}
 \right]\!\right]\right)\right)    
\end{aligned}
\end{equation}
Then, we compute the first and second moments $m^{(\vartheta)} = \frac{\beta_{1}^\alpha m^{(\vartheta-1)} + (1-\beta_{1}^\alpha)g_1^{(\vartheta)}}{1-(\beta_{1}^\alpha)^{\vartheta}}$ and $v^{(\vartheta)} = \frac{\beta_{2}^\alpha v^{(\vartheta-1)} + (1-\beta_{2}^\alpha)\left(g_1^{(\vartheta)}\right)^2}{1-(\beta_{2}^\alpha)^{\vartheta}}$, where $m^{(0)}=0$, $v^{(0)}=0$, and $\left\{\beta_{1}^\alpha,\beta_{2}^\alpha\right\}$ are the hyperparameters. The $\hat{\alpha}_1^{(\vartheta)} = \alpha_{1,\mathrm{opt}}^{(\vartheta-1)} - \varrho_\alpha\frac{m^{(\vartheta)}}{\sqrt{v^{(\vartheta)}} + \varepsilon_\alpha}$ is obtained, where $\varepsilon_\alpha$ is a small constant and $\varrho_\alpha$ is learning rate.

\textit{Step 5:} Repeat \textit{Steps 3-4} until the $\left|\varsigma^{(\vartheta)}-\varsigma^{(\vartheta-1)} \right|/\varsigma^{(\vartheta-1)} \le 10^{-8}$, where $\varsigma^{(\vartheta)} = \left\|\mathcal{Y}-\left[\!\left[\mathbf{\tilde{A}}^{\vartheta}, \mathbf{\tilde{B}}^{\vartheta},\mathbf{\tilde{C}}^{\vartheta}\right]\!\right]\right\|_{\mathrm{F}}^2$.

With the TP-ALS algorithm, 
$\left\{\hat{\mathbf{A}},\hat{\mathbf{B}},\hat{\mathbf{C}} \right\}$ are obtained for assisting sensing. 
Although these factor matrices 
maintain the coupling of angle, delay, and Doppler information, 
the one-to-one relationship between MPCs, VAs, and targets no longer exists. 
Thus, a subsequent matching operation is needed to reestablish this relationship, 
as detailed in Section~\ref{sce3-B}. 
To avoid ambiguity, $\left\{\hat{\mathbf{A}},\hat{\mathbf{B}},\hat{\mathbf{C}} \right\}$ are rewritten as
\begin{equation}\label{eq17}
\begin{aligned}
&\hat{\mathbf{A}}=\left[\hat{\mathbf{a}}_\text{an}\left
(\theta_\xi\right)\right]|_{\xi=1,2,\cdots,I(L+1+K_s)},
\\&
\hat{\mathbf{B}}=\left[\hat{\mathbf{a}}_\text{de}\left
(\tau_\xi\right)\right]|_{\xi=1,2,\cdots,I(L+1+K_s)},
\\&
\hat{\mathbf{C}}=\left[\hat{\mathbf{f}}\left
(f_{\text{D},\xi}\right)\right]|_{\xi=1,2,\cdots,I(L+1+K_s)}.
\end{aligned}
\end{equation}

\subsection{Symbol-Level Fusion Stage}\label{sce3-B}


\begin{figure}[htbp]
    \centering   \includegraphics[width=.38\textwidth]{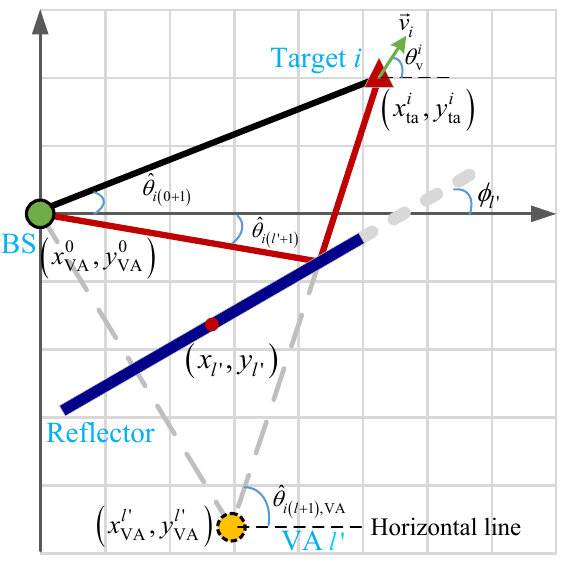}
    \caption{Geometric relationship between VAs and the BS}
    \label{fig4}
\end{figure}

As illustrated in Fig.~\ref{fig4}, 
we define the angle and center coordinate of the $l'\in\{1,2,\cdots, L\}$-th 
reflector as $\phi_{l'}\in\left[0,\pi\right]$ and $(x_{l'}, y_{l'})$, respectively. 
Therefore, the coordinate of the ${l'}$-th VA can be expressed as
\begin{equation}\label{eq18}
 \begin{aligned}
x_{\text{VA}}^{l'}&=\frac{-2\left[y_{l'}-\tan (\phi_{l'}) \cdot x_{l'}\right] \tan (\phi_{l'})}{1+\tan ^{2}(\phi_{l'})}, \\
    y_{\text{VA}}^{l'}&=\frac{2\left(y_{l'}-\tan (\phi_{l'}) \cdot x_{l'}\right)}{1+\tan ^{2}(\phi_{l'})},
\end{aligned}   
\end{equation}
and the coordinate of the BS is denoted by $\left(x_{\text{VA}}^{0},y_{\text{VA}}^{0}\right)$.

\subsubsection{Matching operation}
In Section~\ref{sec3-2}, the separated $I(L+1)$ MPCs are left unordered and 
unassociated with VAs and targets. Furthermore, clutter, such as temporary MPCs, must be suppressed to avoid ghosts, although \cite{shi2022device} indicates that ghosts are impossible if no three VAs are collinear and $K \geq 2I + 1$ holds. 
The matching operation consists of two steps: 
1) From $I(L+1+K_s)$ MPCs, $I(L+1)$ MPCs are selected to be associated with $L$ VAs and the BS, with each VA and the BS receiving an MPC cluster containing $I$ MPCs; 
2) The $I$ MPCs of each VA are associated with $I$ targets.
The detailed descriptions are as follows.

\textit{Step 1:}
Calculate the $L$ angles between $L$ reflectors and the BS, where the 
$l'$-th angle is denoted by $\eta_{l'}$, 
and assume that $\eta_{l'_1}\ne \eta_{l'_2}$, $\forall l'_1 \ne l'_2 \in\{1,2,\cdots,L\} $.
The MUSIC method~\cite{liu2024carrier} is applied to each column of $\hat{\mathbf{A}}$ 
to estimate the angles $\{\hat{\theta}_\xi\}_{\xi=1}^{I(L+1+K_s)}$ of $I(L+1+K_s)$ MPCs.
Then, for each VA, we traverse the angles $\{\hat{\theta}_\xi\}_{\xi=1}^{I(L+1+K_s)}$ 
using the index $\xi\in\{1,2,\cdots,I(L+1+K_s)\}$. 
Taking the $l'$-th VA as an example, 
when $\hat{\theta}_\xi\in\left[\eta_{l'}-\epsilon, \eta_{l'}+\epsilon\right]$, 
the $\xi$-th MPC is considered to match the $l'$-th VA, 
where $\epsilon$ is an angular threshold. 
Finally, the $I$ MPCs that match the $l'$-th VA are reorganized into an MPC cluster matrix, which is denoted by
\begin{equation}\label{eq19}
\resizebox{0.4\textwidth}{!}{$
\Phi_\text{VA}^{l'}=\left[\begin{array}{ccc}
\hat{\mathbf{a}}_\text{an}\left(\theta_{1(l'+1)}\right) & \hat{\mathbf{a}}_\text{de}\left(\tau_{1(l'+1)}\right) & \hat{\mathbf{f}}\left(f_{\text{D},1(l'+1)}\right) \\
\hat{\mathbf{a}}_\text{an}\left(\theta_{2(l'+1)}\right)  & \hat{\mathbf{a}}_\text{de}\left(\tau_{2(l'+1)}\right) & \hat{\mathbf{f}}\left(f_{\text{D},2(l'+1)}\right) \\ 
      \vdots & \ddots & \vdots \\
\hat{\mathbf{a}}_\text{an}\left(\theta_{i'(l'+1)}\right)  & \hat{\mathbf{a}}_\text{de}\left(\tau_{i'(l'+1)}\right) & \hat{\mathbf{f}}\left(f_{\text{D},i'(l'+1)}\right) \\
      \vdots & \ddots & \vdots \\
\hat{\mathbf{a}}_\text{an}\left(\theta_{I(l'+1)}\right)  & \hat{\mathbf{a}}_\text{de}\left(\tau_{I(l'+1)}\right) & \hat{\mathbf{f}}\left(f_{\text{D},I(l'+1)}\right) \\
    \end{array}\right],
$}
\end{equation}  
where $i'\in\{1,2,\cdots,I\}$.
For the BS, the $I$ most powerful MPCs are selected to associated with the BS, and the MPCs cluster matrix of BS is $\Phi_\text{VA}^0$. 
Given that the index $i'$ of MPC in (\ref{eq19}) does not correspond to the target's index $i$, we establish the one-to-one relationship 
between MPCs and targets for each VA and the BS in the \textit{Step 2}.

\textit{Step 2:}
Based on the exhaustive search algorithm in \hyperref[apex1]{\textbf{Appendix A}} and delay estimates $\{\hat{\tau}_{i(l+1)}\}_{i=1}^{I}$ 
obtained by performing the MUSIC method on the $2$-nd column of $\Phi_\text{VA}^{l}$ in the $\{\Phi_\text{VA}^{l}\}_{l=0}^L$, 
the one-to-one relationship between MPCs and targets in $L$ VA and BS is established, 
and the matched MPC cluster matrix of the $l'$-th VA is
\begin{equation}\label{eq20} 
\resizebox{0.4\textwidth}{!}{$
\tilde{\Phi}_\text{VA}^{l'}=\left[\begin{array}{ccc}
\hat{\mathbf{a}}_\text{an}\left(\theta_{1(l'+1)}\right) & \hat{\mathbf{a}}_\text{de}\left(\tau_{1(l'+1)}\right) & \hat{\mathbf{f}}\left(f_{\text{D},1(l'+1)}\right) \\
\hat{\mathbf{a}}_\text{an}\left(\theta_{2(l'+1)}\right)  & \hat{\mathbf{a}}_\text{de}\left(\tau_{2(l'+1)}\right) & \hat{\mathbf{f}}\left(f_{\text{D},2(l'+1)}\right) \\ 
      \vdots & \ddots & \vdots \\
\hat{\mathbf{a}}_\text{an}\left(\theta_{i(l'+1)}\right)  & \hat{\mathbf{a}}_\text{de}\left(\tau_{i(l'+1)}\right) & \hat{\mathbf{f}}\left(f_{\text{D},i(l'+1)}\right) \\
      \vdots & \ddots & \vdots \\
\hat{\mathbf{a}}_\text{an}\left(\theta_{I(l'+1)}\right)  & \hat{\mathbf{a}}_\text{de}\left(\tau_{I(l'+1)}\right) & \hat{\mathbf{f}}\left(f_{\text{D},I(l'+1)}\right) \\
    \end{array}\right].
$}
\end{equation}
Similar to (\ref{eq20}), the matched MPC cluster matrix of the BS is $\tilde{\Phi}_\text{VA}^{0}$. 
So far, the multi-target sensing problem is decomposed 
into $I$ individual single-target sensing problems. 

It is mentioned that the existing state-of-the-art scheme in \cite{gong2022multipath} 
proposed a single-target localization method of data-level fusion. 
To establish a benchmark for the absolute velocity estimation method proposed in this paper, 
we have developed an absolute velocity estimation method of data-level fusion in \hyperref[apex2]{\textbf{Appendix B}}, 
drawing inspiration from the methodology outlined in \cite{gong2022multipath}.
Section \ref{se5} verifies that the SL-MPS scheme is superior to 
the existing state-of-the-art schemes in both localization and absolute velocity estimation.

\subsubsection{Target localization method}
The target's location can be determined by the delay and angle information. 
Given the general scarcity of antenna resources, 
the angles information are utilized solely to narrow the searching scope.
The detailed localization method for the $i$-th target is outlined below.

\textit{Step 1}: 
Based on the $\{\hat{\theta}_{i(l+1),\text{VA}}\}_{l=0}^{L}$ obtained in \hyperref[apex2]{\textbf{Appendix B}}, 
a location searching scope with a length and width of $\Omega\Delta R$ is obtained. 
Then, the searching scope is gridded with a grid size of $\Delta R$, resulting in a searching matrix $\mathbf{P}\in\mathbb{C}^{\Omega\times\Omega}$, 
where $\left[\mathbf{P}\right]_{\rho ,\varsigma_\text{r}}=\left(x_{\rho ,\varsigma_\text{r}}^\text{g},y_{\rho ,\varsigma_\text{r}}^\text{g}\right)$ represents the candidate coordinate with $\rho ,\varsigma_\text{r} \in \{1,2,\cdots,\Omega\}$.

\textit{Step 2}: 
Based on $\mathbf{P}$, $L+1$ delay candidate vectors denoted by $\{\vec{\mathbf{D}}_l\}_{l=0}^{L}$ are obtained. 
The $l$-th delay candidate vector is $\vec{\mathbf{D}}_l\in\mathbb{C}^{1\times\Omega^2}$, 
where the $\xi_\text{r}$-th element of $\vec{\mathbf{D}}_l$ is 
\begin{equation}\label{eq21}
\scalebox{1}{$
\vec{\mathbf{D}}_l\left(\xi_\text{r}\right)=\frac{2\sqrt{\left(y_{\text{VA}}^{l}-\left[\text{vec}(\mathbf{P})\right]_{\xi_\text{r}}^\text{T}\left \langle 2 \right \rangle\right)^2+\left(x_{\text{VA}}^l-\left[\text{vec}(\mathbf{P})\right]_{\xi_\text{r}}^\text{T}\left \langle 1 \right \rangle\right)^2}}{c_0},
$}
\end{equation}
and $\xi_\text{r}\in\{1,2,\cdots,\Omega^2\}$.

\textit{Step 3}: 
Based on $\{\vec{\mathbf{D}}_l\}_{l=0}^{L}$ and (\ref{eq9}), 
the $L+1$ delay matching matrices denoted by $\{\mathbf{D}_l\}_{l=0}^{L}$ are obtained, 
where the $l$-th delay matching matrix $\mathbf{D}_l\in\mathbb{C}^{N_\text{c}\times \Omega^2}$ is expressed as
\begin{equation}\label{eq22}
\scalebox{0.75}{$
  \mathbf{D}_l= \begin{bmatrix}
 1 & 1 & \cdots & 1 \\
 e^{-j2\pi\Delta f\left[\vec{\mathbf{D}}_l\right]_1} & e^{-j2\pi\Delta f\left[\vec{\mathbf{D}}_l\right]_2} & \cdots & e^{-j2\pi\Delta f\left[\vec{\mathbf{D}}_l\right]_{\Omega^2}}\\
 \vdots & \vdots & \ddots & \vdots \\
 e^{-j2\pi\Delta f\left(N_\text{c}-1\right)\left[\vec{\mathbf{D}}_l\right]_1} & e^{-j2\pi\Delta f\left(N_\text{c}-1\right)\left[\vec{\mathbf{D}}_l\right]_2} & \cdots & e^{-j2\pi\Delta f\left(N_\text{c}-1\right)\left[\vec{\mathbf{D}}_l\right]_{\Omega^2}}
\end{bmatrix}.$}
\end{equation}

\textit{Step 4}: 
$\{\mathbf{D}_l\}_{l=0}^{L}$ and corresponding $L+1$ delay vectors 
$\{\hat{\mathbf{a}}_\text{de}\left(\tau_{i(l+1)}\right)\}_{l=0}^{L}$ are fused to obtain a delay fusion vector 
$\mathbf{J}\in\mathbb{C}^{1\times\Omega^2}$, which is expressed as
\begin{equation}\label{weigth_location}
\mathbf{J}=\sum_{l=0}^L\left(\frac{n_{\text{t},i}^l}{\sum_{l=0}^Ln_{\text{t},i}^l}|\mathbf{J}_l|\right),  
\end{equation}
where $\mathbf{J}_l$ is the $l$-th delay profile, expressed as
\begin{equation}\label{eq23}
\mathbf{J}_l=\left[\hat{\mathbf{a}}_\text{de}\left(\tau_{i(l+1)}\right)\right]^\text{T}\mathbf{D}_l^*,
\end{equation}
and $n_{\text{t},i}^l$ denotes the noise variance of $\hat{\mathbf{a}}_\text{de}\left(\tau_{i(l+1)}\right)$, 
which can be obtained by the MLE~\cite{morelli2012joint}.

\textit{Step 5}: 
The peak value of $\mathbf{J}$ is searched to obtain the peak index value $\hat{\xi}_\text{r}$. 
The value of $\left[\text{vec}(\mathbf{P})\right]_{\hat{\xi}_\text{r}}^\text{T}$ represents the 
$i$-th estimated target's location, denoted by $\left(\tilde{x}_\text{ta}^i,\tilde{y}_\text{ta}^i\right)$.

\subsubsection{Target absolute velocity estimation method}
The target's absolute velocity is associated with the angle and Doppler information. 
The absolute velocity estimation method for the $i$-th target is as follows.

\textit{Step 1}: 
Obtain a velocity searching scope with an angle interval of $\left[0, 2\pi\right]$ and magnitude interval of $\left[V_\text{min}, V_\text{max}\right]$, 
based on the sensing demands in the application scenario of ISAC~\cite{liu2024carrier}.

\textit{Step 2}: 
The velocity searching scope is gridded with a grid size $\Delta\theta$ for the angle and a grid size $\Delta V$ for the magnitude, 
yielding a velocity searching matrix $\mathbf{V}\in\mathbb{C}^{A\times B}$, 
where $A=\left \lfloor \frac{2\pi}{\Delta \theta} \right \rfloor $ and $B=\left \lfloor \frac{V_\text{max}-V_\text{min}}{\Delta \theta} \right \rfloor$; 
$\left[\mathbf{V}\right]_{a,b}=(\theta_{a,b}^\text{g},v_{a,b}^\text{g})$ with $a\in\{1,2,\cdots,A\}$ and $b\in\{1,2,\cdots,B\}$. 

\textit{Step 3}: 
Based on $\mathbf{V}$, $L+1$ Doppler candidate vectors denoted by $\{\vec{\mathbf{S}}_l\}_{l=0}^{L}$ are obtained. 
The $l$-th Doppler candidate vector is $\vec{\mathbf{S}}_l\in\mathbb{C}^{1\times(A\times B)}$, 
where the $\xi_\text{v}$-th element of $\vec{\mathbf{S}}_l$ is expressed in (\ref{eq24}) shown at the bottom of this page, 
and $\xi_\text{v}\in\{1,2,\cdots,(A\times B)\}$.
\begin{figure*}[b]
  {\noindent} \rule[-10pt]{18cm}{0.1em}
\begin{equation}\label{eq24}
\vec{\mathbf{S}}_l\left(\xi_\text{v}\right)=\frac{-2f_\text{c}\left[\text{vec}\left(\mathbf{V}\right)\right]_{\xi_\text{v}}^\text{T}\left \langle 2 \right \rangle }{c_0}\left[\cos\left(\hat{\theta}_{i(l+1),\text{VA}}-\left[\text{vec}\left(\mathbf{V}\right)\right]_{\xi_\text{v}}^\text{T}\left \langle 1 \right \rangle \right)\right],
\end{equation}    
\end{figure*}

\textit{Step 4}: 
Based on $\{\vec{\mathbf{S}}_l\}_{l=0}^{L}$ and (\ref{eq10}), 
the $L+1$ Doppler matching matrices denoted by $\{\mathbf{S}_l\}_{l=0}^{L}$ are obtained, 
where the $l$-th Doppler matching matrix $\mathbf{S}_l\in\mathbb{C}^{M\times(A\times B)}$ is expressed as
\begin{equation}\label{eq25}
\scalebox{0.75}{$
  \mathbf{S}_l= \begin{bmatrix}
 1 & 1 & \cdots & 1 \\
 e^{j2\pi T\left[\vec{\mathbf{S}}_l\right]_1} & e^{j2\pi T\left[\vec{\mathbf{S}}_l\right]_2} & \cdots & e^{j2\pi T\left[\vec{\mathbf{S}}_l\right]_{A\times B}}\\
 \vdots & \vdots & \ddots & \vdots \\
 e^{j2\pi(M-1) T\left[\vec{\mathbf{S}}_l\right]_1} & e^{j2\pi (M-1)T\left[\vec{\mathbf{S}}_l\right]_2} & \cdots & e^{j2\pi (M-1)T\left[\vec{\mathbf{S}}_l\right]_{A\times B}}\\
\end{bmatrix}.$}
\end{equation}

\textit{Step 5}: 
$\{\mathbf{S}_l\}_{l=0}^{L}$ and corresponding $L+1$ Doppler vectors $\{\hat{\mathbf{f}}\left(f_{\text{D},i(l+1)}\right)\}_{l=0}^{L}$ 
are fused to obtain a Doppler fusion vector $\mathbf{Q}\in\mathbb{C}^{1 \times (A\times B)}$, 
which is expressed as
\begin{equation}\label{weight_velocity}
 \mathbf{Q}= \sum_{l=0}^L\left(\frac{n_{\text{t},i}^l}{\sum_{l=0}^Ln_{\text{t},i}^l}|\mathbf{Q}_l|\right),  
\end{equation}
where $\mathbf{Q}_l$ is the $l$-th Doppler profile, expressed as follows.
\begin{equation}\label{eq26}
   \mathbf{Q}_l= \left[\hat{\mathbf{f}}\left(f_{\text{D},i(l+1)}\right)\right]^\text{T}\mathbf{S}_l^*.
\end{equation}

\textit{Step 6}: 
The peak value of $\mathbf{Q}$ is searched to obtain the peak index value $\tilde{\xi}_\text{v}$. 
The value of $\left[\text{vec}\left(\mathbf{V}\right)\right]_{\tilde{\xi}_\text{v}}^{\text{T}}$ 
represents the $i$-th estimated target's absolute velocity, 
denoted by $(\tilde{\theta}_\text{v}^{i},|\tilde{\Vec{v}}_i|)$.

\hyperref[algo1]{\textbf{Algorithm 1}} demonstrates the proposed SL-MPS scheme.
\begin{table}[!ht]
\centering
\label{algo1}
\resizebox{0.91\linewidth}{!}{
\setlength{\arrayrulewidth}{1.5pt}
\begin{tabular}{rllll}
\hline
\multicolumn{5}{l}{\textbf{Algorithm 1:} The Proposed SL-MPS Scheme}   \\ \hline
\multirow{-6}{*}{\textbf{Input:} }               & \multicolumn{4}{l}{\begin{tabular}[c]{@{}l@{}}The received three-order tensor $\mathcal{Y}$ in (\ref{eq7});\\ The angles and center coordinates of reflectors\\ $\{\left(x_{l'},y_{l'}\right)\}_{l'=1}^{L}$ and $\{\phi_{l'}\}_{l'=1}^{L}$;\\ The coordinate of the BS $\left(x_\text{VA}^0,y_\text{VA}^0\right)$; \\ The grid sizes of angle and magnitude $\Delta\theta$ and $\Delta V$; \\ The grid size $\Delta R$. \end{tabular}} \\
\multirow{-2}{*}{\textbf{Output:} }               & \multicolumn{4}{l}{\begin{tabular}[c]{@{}l@{}}The estimated $I$ target's locations $\{\left(\tilde{x}_\text{ta}^i,\tilde{y}_\text{ta}^i\right)\}_{i=1}^{I}$;\\
The estimated $I$ target's absolute velocities $\{(\tilde{\theta}_\text{v}^{i},|\tilde{\Vec{v}}_i|)\}_{i=1}^{I}$.\end{tabular}} \\ 
\multicolumn{5}{l}{\textbf{MPC Separation Stage:}} \\
\multirow{-2}{*}{1:}       & \multicolumn{4}{l}{\begin{tabular}[c]{@{}l@{}}Obtain the factor matrices $\left\{\hat{\mathbf{A}},\hat{\mathbf{B}},\hat{\mathbf{C}}\right\}$ \\ by the proposed TP-ALS algorithm;\end{tabular}} \\
2:        & \multicolumn{4}{l}{Output $\hat{\mathbf{A}}$, $\hat{\mathbf{B}}$, and $\hat{\mathbf{C}}$ in (\ref{eq17}).}  \\
\multicolumn{5}{l}{\textbf{Symbol-Level Fusion Stage:}} \\
\multirow{-2}{*}{3:}     & \multicolumn{4}{l}{\begin{tabular}[c]{@{}l@{}}Obtain the coordinates of VAs $\{(x_{\text{VA}}^{l'},y_{\text{VA}}^{l'})\}|_{l'=1}^{L}$ by\\ $\{\left(x_{l'},y_{l'}\right)\}_{l'=1}^{L}$, $\{\phi_{l'}\}_{l'=1}^{L}$, and (\ref{eq18});\end{tabular}}  \\
\multirow{-2}{*}{4:}     & \multicolumn{4}{l}{\begin{tabular}[c]{@{}l@{}}Obtain the the matched MPCs cluster matrices $\{\tilde{\Phi}_\text{VA}^l\}_{l=0}^{L}$ \\ by the proposed matching operation;\end{tabular}}  \\
5:  & \multicolumn{4}{l}{\textbf{For} $i$ to $I$ \textbf{do}}  \\
6:                              & \multicolumn{4}{l}{$\hspace{0.8em}$ Obtain a searching matrix $\mathbf{P
}$ by $\Delta R$;}  \\
7:                              & \multicolumn{4}{l}{$\hspace{0.8em}$ Obtain $\{\vec{\mathbf{G}}_l\}_{l=0}^{L}$ by $\mathbf{P
}$ and (\ref{eq21});}  \\
8:                           & \multicolumn{4}{l}{$\hspace{0.8em}$ Obtain $\{\mathbf{D}_l\}_{l=0}^{L}$ by $\{\vec{\mathbf{G}}_l\}_{l=0}^{L}$ and (\ref{eq22}); }   \\
9:                              & \multicolumn{4}{l}{$\hspace{0.8em}$ Obtain a $\mathbf{J}$ by $\{\hat{\mathbf{a}}_\text{de}\left(\tau_{i(l+1)}\right)\}_{l=0}^{L}$, (\ref{weigth_location}), and (\ref{eq23}) ;}  \\
10:                            & \multicolumn{4}{l}{$\hspace{0.8em}$ Obtain $\hat{\xi}_\text{r}$ from $\mathbf{J}$ to find the $i$-th target's location $\left(\tilde{x}_\text{ta}^i,\tilde{y}_\text{ta}^i\right)$;}  \\
11:                              & \multicolumn{4}{l}{$\hspace{0.8em}$ Obtain a searching matrix $\mathbf{V
}$ by $\Delta \theta$ and $\Delta V$ ;}  \\
12:                              & \multicolumn{4}{l}{$\hspace{0.8em}$ Obtain $\{\vec{\mathbf{S}}_l\}_{l=0}^{L}$ by  $\mathbf{V}$ and (\ref{eq24});}  \\
13:                           & \multicolumn{4}{l}{$\hspace{0.8em}$ Obtain $\{\mathbf{S}_l\}_{l=0}^{L}$ by $\{\vec{\mathbf{S}}_l\}_{l=0}^{L}$ and (\ref{eq25}); }   \\
14:                              & \multicolumn{4}{l}{$\hspace{0.8em}$ Obtain a $\mathbf{Q}$ by $\{\mathbf{S}_l\}_{l=0}^{L}$, $\{\hat{\mathbf{f}}\left(f_{\text{D},i(l+1)}\right)\}_{l=0}^{L}$, (\ref{weight_velocity}), and (\ref{eq26}) ;}  \\
\multirow{-2}{*}{15:}     & \multicolumn{4}{l}{\begin{tabular}[c]{@{}l@{}}$\hspace{0.8em}$ Obtain $\hat{\xi}_\text{v}$ from $\mathbf{Q}$ to find the $i$-th target's \\ $\hspace{0.8em}$ absolute velocity $(\tilde{\theta}_\text{v}^{i},|\tilde{\Vec{v}}_i|)$;\end{tabular}}  \\
16:  & \multicolumn{4}{l}{\textbf{End} \textbf{For}}
\\
\hline
\end{tabular}}
\end{table}

\section{Performance Analysis}\label{se4}
In this section, 
the CRLBs of localization and absolute velocity estimation with MPC information are derived.

For an estimate vector parameter ${\beta}=\left[\beta_1,\beta_2,\cdots,\beta_w\right]$, 
the CRLB can be expressed as~\cite{godrich2010target}
\begin{equation}\label{eq27}
    \text{CRLB}_\beta=\left[\mathbf{J}(\beta)\right]^{-1},
\end{equation}
where $\mathbf{J}(\cdot) \in \mathbb{R}^{w\times w}$ is a Fisher information matrix (FIM). The $(c,j)$-th element of $\mathbf{J}(\beta)$ is~\cite{godrich2010target}
\begin{equation}\label{eq28}
    \left[\mathbf{J}(\beta)\right]_{c,j}=-E\left[\frac{\partial^2\ln{p(r;\beta)}}{\partial \beta_c\partial \beta_j}\right],\quad c,j\in\{1,2,\cdots,w\}
\end{equation}
where $r$ is the received signal and $p(\cdot)$ denotes the probability density function.

For an estimate $\varphi =g(\beta)\in \mathbb{R}^{q\times 1}$, 
the CRLB can be expressed as~\cite{godrich2010target}
\begin{equation}\label{eq29}
    \text{CRLB}_\varphi=\left[\frac{\partial \beta}{\partial \varphi}\mathbf{J}(\beta)\left[\frac{\partial \beta}{\partial \varphi}\right]^\text{T}\right]^{-1},
\end{equation}
where $\frac{\partial\beta}{\partial \varphi}\in \mathbb{R}^{q\times w}$ is the Jacobian matrix.

\subsection{CRLB of Localization}
For the $i$-th target, 
we assume that the estimate vector parameter is $\mathbf{p}=\left[x_\text{ta}^i,y_\text{ta}^i\right]^\text{T}$, 
and the CRLB of localization is revealed in \hyperref[theo1]{Theorem 1}.

\begin{theorem}\label{theo1}
\textit{ Based on (\ref{eq29}), the CRLB of localization is expressed as
 \begin{equation}\label{eq31}
     \text{CRLB}_{\mathbf{p}}=\left[\mathbf{P}\mathbf{F}(\mathbf{c})\mathbf{P}^\text{T}\right]^{-1},
 \end{equation}
 where $\mathbf{P}\in\mathbb{R}^{2\times 3L}$ is a Jacobian matrix expressed in (\ref{eq47}), 
 and $\mathbf{F}(\mathbf{c})\in \mathbb{R}^{3L \times 3L}$ is an FIM expressed in (\ref{eq38}).
Upon observing (\ref{eq31}), 
the CRLB of localization changes with the locations and absolute velocities of targets, 
and the locations of VAs. }
\end{theorem}
\begin{proof}
For a MIMO-OFDM ISAC system, 
we assume that there exist $L$ MPCs about the $i$-th target. 
Each MPC is assumed to be a signal emitted by a VA and the received ISAC echo signal of the $l'$-th VA is
\begin{equation}\label{eq32}
\begin{aligned}
y_i^{l'}=&\sum_{k=1}^{N_\text{R}}\sum_{n=1}^{N_\text{c}}\sum_{m=1}^{M}\left(\alpha_i^{l'} e^{j2\pi mf_{\text{D},i}^{l'}T}e^{-j2\pi n\Delta f\tau_i^{l'}} \right. \\ & \left. \times e^{j2\pi\frac{\text{d}_r}{\lambda} k\sin(\theta_{i,\text{VA}}^{l'})}+z_i^{l'}\right),   
\end{aligned}
\end{equation}
where $z_i^{l'} \sim \mathcal{CN}(0,\sigma_z^2)$;
$\theta_{i,\text{VA}}^{l'}$ denotes the AoA between the $l'$-th VA and the $i$-th target.

Observing Eqs. (\ref{eq29}) and (\ref{eq32}), 
the target location is related to the parameters in the received echo signal. 
Thus, we first derive the FIM of the parameters, 
and then obtain the CRLB of localization.

We define an unknown parameter vector $\mathbf{c}=\left[\mathbf{e}, \mathbf{g}, \mathbf{h}\right]$, where
\begin{equation}\label{eq33}
  \begin{aligned} 
\mathbf{e} & = \left[\tau_i^1,\tau_i^2,\cdots\tau_i^L\right], \\
 \mathbf{g} & = \left[\sin\left(\theta_{i,\text{VA}}^1\right),\sin\left(\theta_{i,\text{VA}}^2\right),\cdots,\sin\left(\theta_{i,\text{VA}}^L\right)\right], \\
\mathbf{h} & = \left[f_{\text{D},i}^1,f_{\text{D},i}^2,\cdots,f_{\text{D},i}^L\right]. 
\end{aligned}  
\end{equation}
According to (\ref{eq32}), the log-likelihood function of $\mathbf{c}$ is 
\begin{equation}\label{eq34}
\ln p\left(\mathbf{y};\mathbf{c}\right)\propto-\frac{1}{2\sigma_z^2} \sum\limits_{l'}\left|y_i^{l'}-\sum\limits_k \sum\limits_n \sum\limits_ms_{l',k,n,m}^i\right|^2.   
\end{equation}
where $\mathbf{y}=[y_i^1,y_i^2,\cdots,y_i^L]^\text{T}$ denotes the vector form of $L$ received echo signals, and 
\begin{equation}\label{eq35}
    s_{l',k,n,m}^i=\alpha_i^{l'}e^{j2\pi mf_{\text{D},i}^{l'}T}e^{-j2\pi n\Delta f\tau_i^{l'}}e^{j2\pi\frac{\text{d}_r}{\lambda} k\sin(\theta_{i,\text{VA}}^{l'})}.
\end{equation}

According to (\ref{eq28}) and (\ref{eq34}), 
the matrix $\mathbf{F}(\mathbf{c})$ can be represented as the form of block matrix, 
denoted by
\begin{equation}\label{eq38}
\mathbf{F}(\mathbf{c})=\begin{bmatrix}\mathbf{D}_{L\times L} &\mathbf{L}_{L\times L} & \mathbf{S}_{L\times L} \\ \mathbf{L}_{L\times L}&\mathbf{A}_{L\times L}&\mathbf{E}_{L\times L} \\ \mathbf{S}_{L\times L}&\mathbf{E}_{L\times L}&\mathbf{R}_{L\times L}\end{bmatrix},
\end{equation}
where $\mathbf{D}$, $\mathbf{L}$, $\mathbf{S}$, $\mathbf{A}$, $\mathbf{E}$ and $\mathbf{R}$ are all the diagonal matrices. 
The $(l',l')$-th elements of $\mathbf{D}$, $\mathbf{L}$, $\mathbf{S}$, $\mathbf{A}$, $\mathbf{E}$ and $\mathbf{R}$ 
can be expressed in (\ref{eq39}), (\ref{eq40}), (\ref{eq41}), (\ref{eq42}), (\ref{eq43}), and (\ref{eq44}), respectively.
\begin{figure*}
\centering
\begin{minipage}{0.48\textwidth}
\begin{equation}\label{eq39}
{\fontsize{9}{9}\begin{aligned}
 \left[\mathbf{D}\right]_{l',l'} &=-E\left[\frac{\partial^2\ln p(\mathbf{y};\mathbf{c})}{\partial^2\tau_i^{l'}}\right]\\ &=\frac{4\pi^2\left|\alpha_i^{l'}\right|^2\Delta f^2(2N_\text{c}+1)(N_\text{c}+1)N_\text{c}N_\text{R}M}{6\sigma^2},   
\end{aligned}}   
\end{equation}
\end{minipage}
\hfill
\begin{minipage}{0.48\textwidth}
 \begin{equation}\label{eq40}
{\fontsize{9}{9} \begin{aligned}
 \left[\mathbf{L}\right]_{l',l'}&=-E\left[\frac{\partial^2\ln p\left(\mathbf{y};\mathbf{c}\right)}{\partial\tau_i^{l'}\partial\sin\left(\theta_{i,\text{VA}}^{l'}\right)}\right]\\&=\frac{-4\pi^2\left|\alpha_i^{l'}\right|^2\frac{\text{d}_r}{\lambda}\Delta f(N_\text{R}+1)N_\text{R}(N_\text{c}+1)N_\text{c}M}{4\sigma^2},     
 \end{aligned}}
\end{equation} 
\end{minipage}
\end{figure*}
\begin{figure*}
    \begin{minipage}{0.48\textwidth}
      \begin{equation}\label{eq41}
 {\fontsize{9}{9}\begin{aligned}
 \left[\mathbf{S}\right]_{l',l'}&=-E\left[\frac{\partial^2\ln p\left(\mathbf{y};\mathbf{c}\right)}{\partial\tau_i^{l'}\partial f_{\text{D}, i}^{l'}}\right]\\&=\frac{-4\pi^2\left|\alpha_i^{l'}\right|^2T\Delta f(M+1)M(N_\text{c}+1)N_\text{c}N_\text{R}}{4\sigma^2},     
 \end{aligned}}
\end{equation}  
    \end{minipage}
    \hfill
    \begin{minipage}{0.48\textwidth}
       \begin{equation}\label{eq42}
      {\fontsize{9}{9} \begin{aligned}
          \left[\mathbf{A}\right]_{{l'},{l'}}&=-E\left[\frac{\partial^2\ln p\left(\mathbf{y};\mathbf{c}\right)}{\partial^2\sin\left(\theta_{i,\text{VA}}^{l'}\right)}\right]\\&=\frac{4\pi^2\left|\alpha_i^{l'}\right|^2\left(\frac{\text{d}_r}{\lambda}\right)^2(2N_\text{R}+1)(N_\text{R}+1)N_\text{R}N_\text{c}M}{6\sigma^2}. 
       \end{aligned}}
\end{equation} 
    \end{minipage}
\end{figure*}
\begin{figure*}
    \begin{minipage}{0.48\textwidth}
      \begin{equation}\label{eq43}
      \begin{aligned}
          \left[\mathbf{E}\right]_{{l'},{l'}}&=-E\left[\frac{\partial^2\ln p\left(\mathbf{y};\mathbf{v}\right)}{\partial\sin\left(\theta_{i,\text{VA}}^{l'}\right)\partial f_{\text{D},i}^{l'}}\right]\\& =\frac{4\pi^2\left|\alpha_i^{l'}\right|^2\frac{\text{d}_r}{\lambda}T(N_\text{R}+1)N_\text{R}(M+1)MN_\text{c}}{4\sigma^2}, 
      \end{aligned}
    \end{equation}  
    \end{minipage}
    \hfill
    \begin{minipage}{0.48\textwidth}
      \begin{equation}\label{eq44}
      \begin{aligned}
        \left[\mathbf{R}\right]_{{l'},{l'}}&=-E\left[\frac{\partial^2\ln p\left(\mathbf{y};\mathbf{v}\right)}{\partial^2f_{\text{D},i}^{l'}}\right]\\&=\frac{4\pi^2\left|\alpha_i^{l'}\right|^2T^2(2M+1)(M+1)MN_\text{R}N_\text{c}}{6\sigma^2}.  
      \end{aligned}     
    \end{equation}  
    \end{minipage}
 {\noindent} \rule[-10pt]{18cm}{0.1em}
\end{figure*}

With the derived FIM of $\mathbf{c}$, 
the CRLB of localization can be derived as follows. 
For the estimate vector $\mathbf{p}$ and $\mathbf{c}$, 
we have the following relationship.
\begin{equation}\label{eq46}
\begin{cases}\tau_i^{l'}=\frac{2\sqrt{\left(x_{\text{VA}}^{l'}-x_{\text{ta}}^i\right)^2+\left(y_{\text{VA}}^{l'}-y_{\text{ta}}^i\right)^2}}{c_0}\\ y_{\text{ta}}^i-y_{\text{VA}}^{l'}=\tan\left(\theta_{i,\text{VA}}^{l'}\right)\left(x_{\text{ta}}^i-x_{\text{VA}}^{l'}\right) \\
f_{\text{D},i}^{l'}=\frac{-2f_{\text{c}}\left|\vec{v}_i\right|}{c_0}\cos\left(\theta_{i,\text{VA}}^{l'}-\theta_{\text{v}}^i\right)
\end{cases}.
\end{equation}

According to (\ref{eq29}), (\ref{eq38}), and (\ref{eq46}), 
the CRLB of localization can be expressed in (\ref{eq31}), 
and the $\mathbf{P}$ is expressed in (\ref{eq47}) shown at the bottom of the next page, 
where $\eta_i^{l'}=\sqrt{\left(x_{\text{ta}}^i-x_{\text{VA}}^{l'}\right)^2+\left(y_{\text{ta}}^i-y_{\text{VA}}^{l'}\right)^2}$, $\epsilon_i^{l'}=\left(y_{\text{ta}}^i-y_{\text{VA}}^{l'}\right)\left(x_{\text{ta}}^i-x_{\text{VA}}^{l'}\right)$, $\Xi_i^{l'}=f_{\text{D},i}^{l'}
\tan\left(\gamma_i^{l'}-\theta_\text{v}^i\right)$, and $\gamma_i^{l'}=\arctan\left(\frac{y_{\text{ta}}^i-y_{\text{VA}}^{l'}}{x_{\text{ta}}^i-x_{\text{VA}}^{l'}}\right)$.

\renewcommand{\arraystretch}{1.8} 
\begin{figure*}[b]
   {\noindent} \rule[-10pt]{18cm}{0.1em}  \begin{equation}\label{eq47}
    \begin{aligned}
         \mathbf{P}
         =\left[\begin{array}{c@{\hspace{1pt}}c@{\hspace{1pt}}cc@{\hspace{1pt}}c@{\hspace{1pt}}cc@{\hspace{1pt}}c@{\hspace{1pt}}cc@{\hspace{1pt}}c@{\hspace{1pt}}c}
         \frac{-2\left(x_{\text{VA}}^1-x_{\text{ta}}^i\right)} {\eta_i^{1}c_0} & 
         \frac{-2\left(x_{\text{VA}}^2-x_{\text{ta}}^i\right)} {\eta_i^{2}c_0} &
         \cdots &
         \frac{-2\left(x_{\text{VA}}^L-x_{\text{ta}}^i\right)} {\eta_i^{L}c_0} &
         \frac{\epsilon_i^{1}}{-(\eta_i^{1})^3} &
         \frac{\epsilon_i^{2}}{-(\eta_i^{2})^3} &
         \cdots&
         \frac{\epsilon_i^{L}}{-(\eta_i^{L})^3} &
         \frac{\Xi_i^{1}\left(y_{\text{ta}}^i-y_{\text{VA}}^1\right)}{(\eta_i^{1})^2}&
         \frac{\Xi_i^{2}\left(y_{\text{ta}}^i-y_{\text{VA}}^2\right)}{(\eta_i^{2})^2}&
         \cdots&
         \frac{\Xi_i^{L}\left(y_{\text{ta}}^i-y_{\text{VA}}^L\right)}{(\eta_i^{L})^2} 

         \\ \frac{-2\left(y_{\text{VA}}^1-y_{\text{ta}}^i\right)}{\eta_i^{1}c_0}&
         \frac{-2\left(y_{\text{VA}}^2-y_{\text{ta}}^i\right)}{\eta_i^{2}c_0}&
         \cdots&
         \frac{-2\left(y_{\text{VA}}^L-y_{\text{ta}}^i\right)}{\eta_i^{L}c_0}&
         \frac{\left(x_{\text{ta}}^i-x_{\text{VA}}^1\right)^{2}}{\left(\eta_i^{1}\right)^{3}}&
         \frac{\left(x_{\text{ta}}^i-x_{\text{VA}}^2\right)^{2}}{\left(\eta_i^{2}\right)^{3}}&
         \cdots&
         \frac{\left(x_{\text{ta}}^i-x_{\text{VA}}^L\right)^{2}}{\left(\eta_i^{L}\right)^{3}}&
         \frac{-\Xi_i^{1}}{\left(x_{\text{ta}}^i-x_{\text{VA}}^1\right)}&
         \frac{-\Xi_i^{2}}{\left(x_{\text{ta}}^i-x_{\text{VA}}^2\right)}&
         \cdots&
         \frac{-\Xi_i^{L}}{\left(x_{\text{ta}}^i-x_{\text{VA}}^L\right)}\\
        \end{array}\right],
    \end{aligned}    
    \end{equation}
\end{figure*}
\end{proof}

\subsection{CRLB of Absolute Velocity Estimation}
Assume that the estimated vector parameter is $\mathbf{t}=\left[\left|\vec{v}_i\right|,\theta_{\text{v}}^i\right]^\text{T}$, 
and the CRLB of absolute velocity estimation is revealed in \hyperref[theo2]{Theorem 2}.

\begin{theorem}\label{theo2}
\textit{Similar to (\ref{eq31}), 
the CRLB of absolute velocity estimation is represented as 
\begin{equation}\label{eq48}
    \text{CRLB}_\mathbf{t}=\left[\mathbf{T}\mathbf{F}(\mathbf{c})\mathbf{T}^\text{T}\right]^{-1},
\end{equation}
where $\mathbf{T}\in\mathbb{R}^{2\times 3L}$ is a Jacobian matrix expressed in (\ref{eq51}) shown at the bottom of the next page.
Similar to the CRLB of localization, 
the CRLB of absolute velocity also changes with the locations and absolute velocities of targets, 
as well as the locations of VAs. } 
\end{theorem}
\begin{proof}
The CRLB for absolute velocity estimation follows the same process as localization, 
differing only in the separate Jacobian matrices for $\mathbf{t}$ and $\mathbf{p}$. 
Thus, the CRLB of absolute velocity is derived from the Jacobian of $\mathbf{t}$ alone.

For the estimate vector $\mathbf{t}$ and $\mathbf{v}$, 
the relationship in (\ref{eq46}) holds.
Therefore, based on (\ref{eq29}) and (\ref{eq46}), 
the Jacobian matrix $\mathbf{T}$ of $\mathbf{t}$ is derived and detailed in (\ref{eq51}).
\renewcommand{\arraystretch}{1.8} 
\begin{figure*}[b]
    \begin{equation}\label{eq51}
    \begin{aligned}
        \mathbf{T}
=\left[
\begin{array}{c@{\hspace{1pt}}c@{\hspace{1pt}}cc@{\hspace{1pt}}c@{\hspace{1pt}}cc@{\hspace{1pt}}c@{\hspace{1pt}}cc@{\hspace{1pt}}c@{\hspace{1pt}}c}
\frac{\epsilon_i^1f_{\text{D},i}^1}{\eta_i^1c_0\Xi_i^1}&\frac{\epsilon_i^2f_{\text{D},i}^2}{\eta_i^2c_0\Xi_i^2}&\cdots&\frac{\epsilon_i^Lf_{\text{D},i}^L}{\eta_i^Lc_0\Xi_i^L}&
\frac{\cos^2\left(\gamma_i^1-\theta_\text{v}^i\right)}{\left|\vec{v}_i\right|\sin\left(\gamma_i^1-\theta_\text{v}^i\right)}&\frac{\cos^2\left(\gamma_i^2-\theta_\text{v}^i\right)}{\left|\vec{v}_i\right|\sin\left(\gamma_i^2-\theta_\text{v}^i\right)}&\cdots&\frac{\cos^2\left(\gamma_i^L-\theta_\text{v}^i\right)}{\left|\vec{v}_i\right|\sin\left(\gamma_i^L-\theta_\text{v}^i\right)}&
\frac{\cos\left(\gamma_i^1-\theta_\text{v}^i\right)}{\frac{1}{-2f_\text{c}}c_0}&
\frac{\cos\left(\gamma_i^2-\theta_\text{v}^i\right)}{\frac{1}{-2f_\text{c}}c_0}&\cdots&
\frac{\cos\left(\gamma_i^L-\theta_\text{v}^i\right)}{\frac{1}{-2f_\text{c}}c_0} \\
\frac{\epsilon_i^1}{2\eta_i^1c_0}&
\frac{\epsilon_i^2}{2\eta_i^2c_0}&\cdots&
\frac{\epsilon_i^L}{2\eta_i^Lc_0}&
\cos\left(\gamma_i^1\right)&
\cos\left(\gamma_i^2\right)&\cdots&
\cos\left(\gamma_i^L\right)&
\frac{\sin\left(\gamma_i^1-\theta_\text{v}^i\right)}{\frac{1}{-2f_\text{c}\left|\vec{v}_i\right|}c_0}&
\frac{\sin\left(\gamma_i^2-\theta_\text{v}^i\right)}{\frac{1}{-2f_\text{c}\left|\vec{v}_i\right|}c_0}&\cdots&
\frac{\sin\left(\gamma_i^L-\theta_\text{v}^i\right)}{\frac{1}{-2f_\text{c}\left|\vec{v}_i\right|}c_0}
\end{array}\right],
    \end{aligned}
\end{equation}
\end{figure*}
\end{proof}


\section{Simulation}\label{se5}

In this section, 
the proposed SL-MPS scheme is simulated to verify the superiority and feasibility of target localization and absolute velocity estimation. 
The simulation results are obtained with 5000 times Monte Carlo simulations. 
The global parameter setting is listed as follows. 

The carrier frequency is $f_\text{c} = 3.5$ GHz\cite{3gpp2018nr}, 
and subcarrier spacing is $\Delta f = 30$ kHz~\cite{3gpp2018nr}. 
$N_\text{c}$, $M_\text{sym}$, $N_\text{R}$, and $N_\text{T}$ are set to $1024$, $140$, $128$, and $128$, respectively~\cite{liu2024carrier}.
The coordinates of targets are (100, 50) m, (100, 30) m, and (100, 10) m, respectively. 
The absolute velocities of targets are [150 m/s, 27.6°], [140 m/s, 25.7°], and [130 m/s, 24°], respectively.

\subsection{Proposed TP-ALS algorithm}\label{se5-A}
We simulate the normalized mean-squared-error (NMSE) of the proposed TP-ALS algorithm with ALS with random initialization (R-ALS)~\cite{gong2024}, ALS with SVD-based initialization (SVD-ALS)~\cite{gong2024}, R-ALS with linear search (R-ALS+LS)~\cite{kolda2009tensor},
SVD-ALS with linear search (SVD-ALS+LS)~\cite{kolda2009tensor}, and ALS with deep-learning-based initialization (DL-ALS)~\cite{ gong2024} as the benchmark schemes.

Fig.~\ref{fig6} shows that the proposed TP-ALS algorithm has faster convergence, smaller reconstruction error and higher decomposition accuracy at the same SNR.

\begin{figure}[htbp]
	\centering
	\subfigure[NMSE of different algorithms at SNR = 0 dB, varying with the number of iterations] {\label{fig6.a}\includegraphics[width=.31\textwidth]{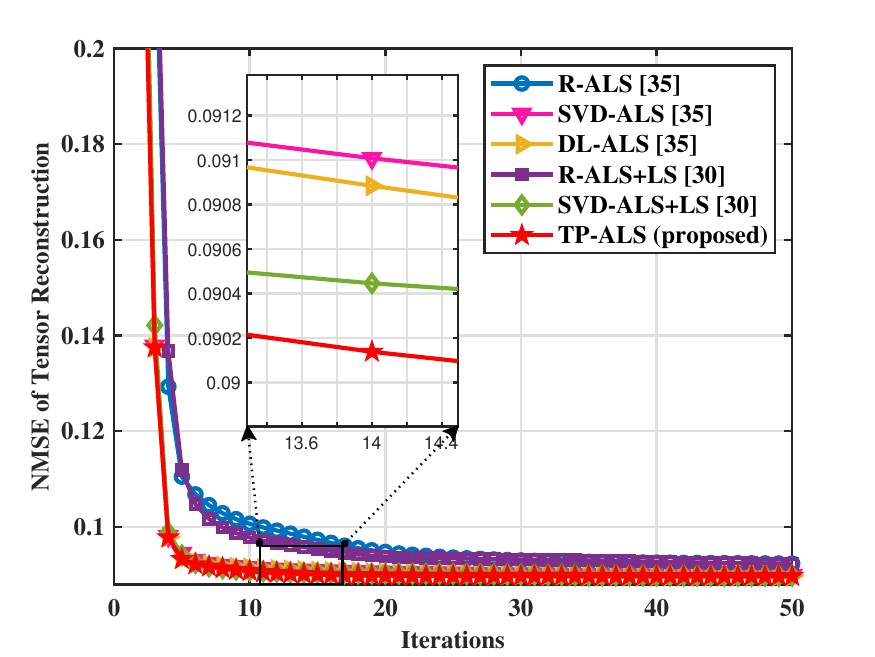}}
	\subfigure[NMSE of different algorithms at iterations = 50, varying with the SNR] {\label{fig6.b}\includegraphics[width=.31\textwidth]{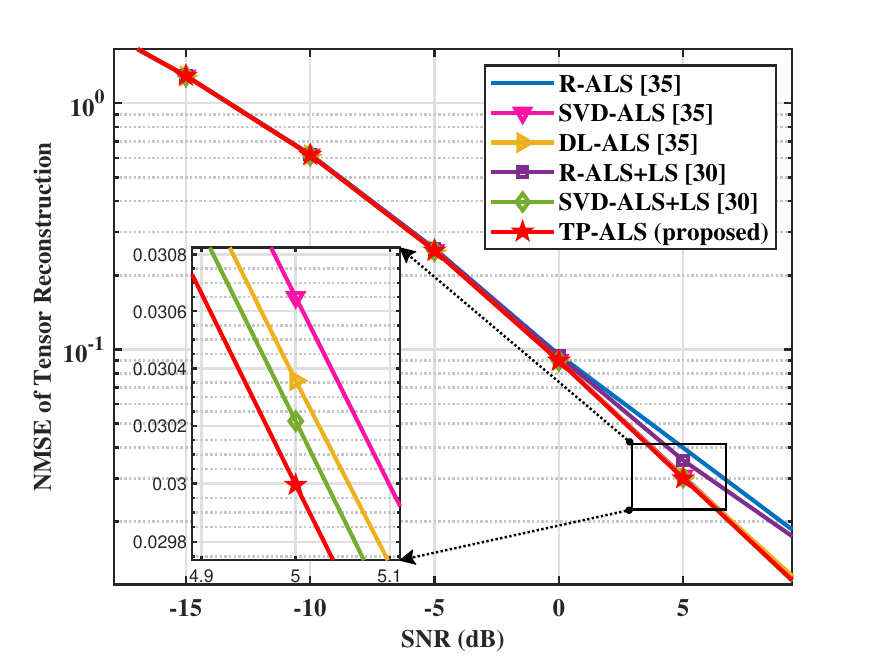}}
	\caption{NMSE comparison of tensor reconstruction for different algorithms}
	\label{fig6}
\end{figure}

\subsection{Sensing Performance}
First, the CRLB is used to evaluate the theoretical sensing performance of the proposed SL-MPS scheme, 
while the average RMSE (ARMSE) is defined to characterize the actual sensing performance. 

\subsubsection{Analysis of CRLB}
Fig. \ref{fig8} demonstrates the CRLBs of location and absolute velocity estimations with varying subcarrier spacing. 
As subcarrier spacing increases, 
the CRLB for location estimation decreases, 
while the CRLB for velocity estimation increases. 
This is because a larger subcarrier spacing increases signal bandwidth, 
improving location accuracy, 
but reduces signal duration, decreasing velocity accuracy. 
Thus, a trade-off exists in selecting the parameter of subcarrier spacing.

Fig.~\ref{fig9} illustrates the CRLBs of location and absolute velocity estimations with varying $L$. 
The $L=0$ refers to the traditional ISAC sensing scheme without the enhancement from MPCs. 
As $L$ increases,
the CRLBs of both location and absolute velocity estimation decrease,
indicating that MPC information can theoretically enhance the sensing accuracy.
\begin{figure}
    \centering
\includegraphics[width=0.30\textwidth]{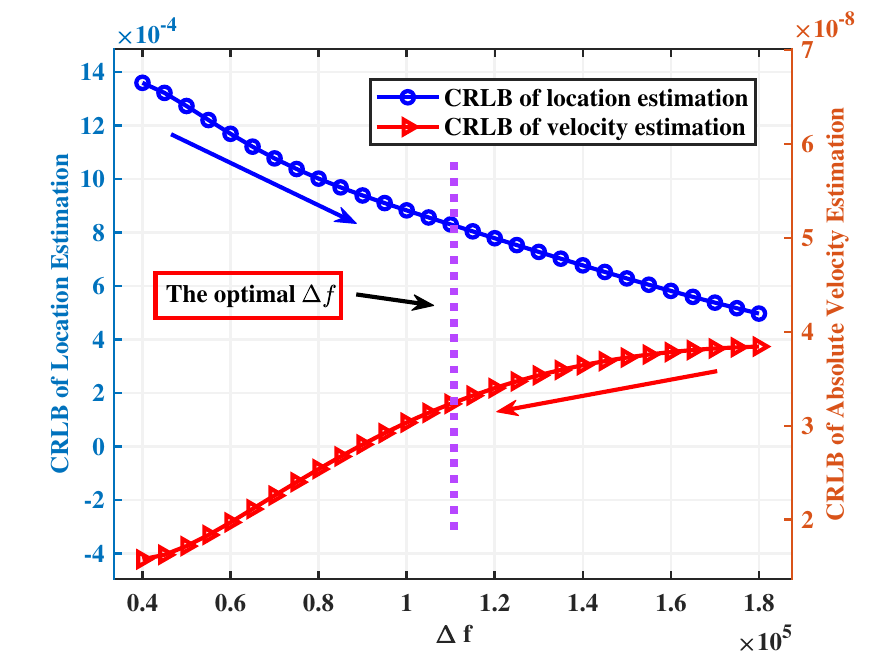}
    \caption{CRLBs of location and velocity estimations with SNR= -10 dB}
    \label{fig8}
\end{figure}

\begin{figure}[htbp]
	\centering
	\subfigure[Location estimation] {\label{fig9.a}\includegraphics[width=.24\textwidth]{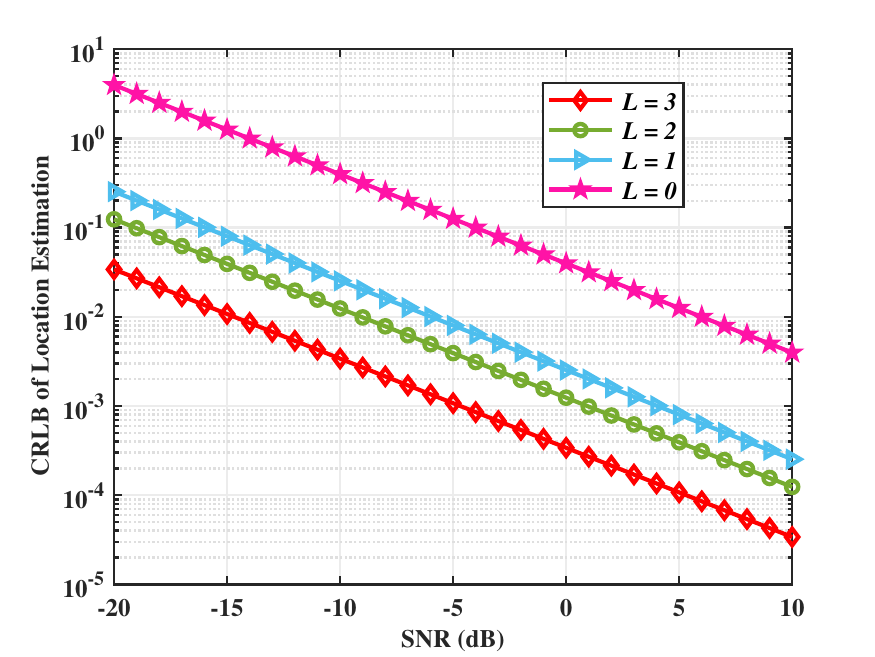}}
	\subfigure[Absolute velocity estimation] {\label{fig9.b}\includegraphics[width=.24\textwidth]{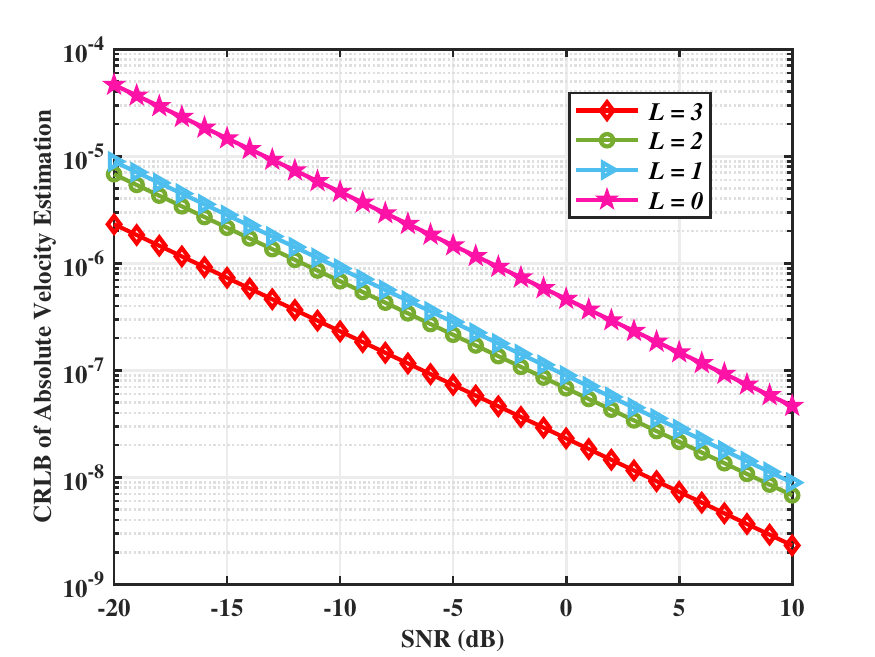}}
	\caption{CRLBs of location and absolute velocity estimations}
	\label{fig9}
\end{figure}

\subsubsection{Sensing profiles}
To validate the feasibility of the proposed SL-MPS scheme, 
we first simulate the sensing profiles for target localization and absolute velocity estimation with different $L$. 
Since the multi-target sensing problem is decomposed into $L$ single-target problems in Section~\ref{sce3-B}, 
we confine our simulation to a single target without loss of generality.

Figs. \ref{fig10} and \ref{fig11} illustrate the sensing profiles with different $L$.
When $L=0$, the profile lacks a distinct peak due to only one LoS MPC. 
As $L$ increases, the energy of the profile peak becomes more concentrated, 
enhancing sensing accuracy.
\begin{figure}[htbp]
	\centering
	\subfigure[$L = 0$] {\label{fig10.a}\includegraphics[width=.22\textwidth]{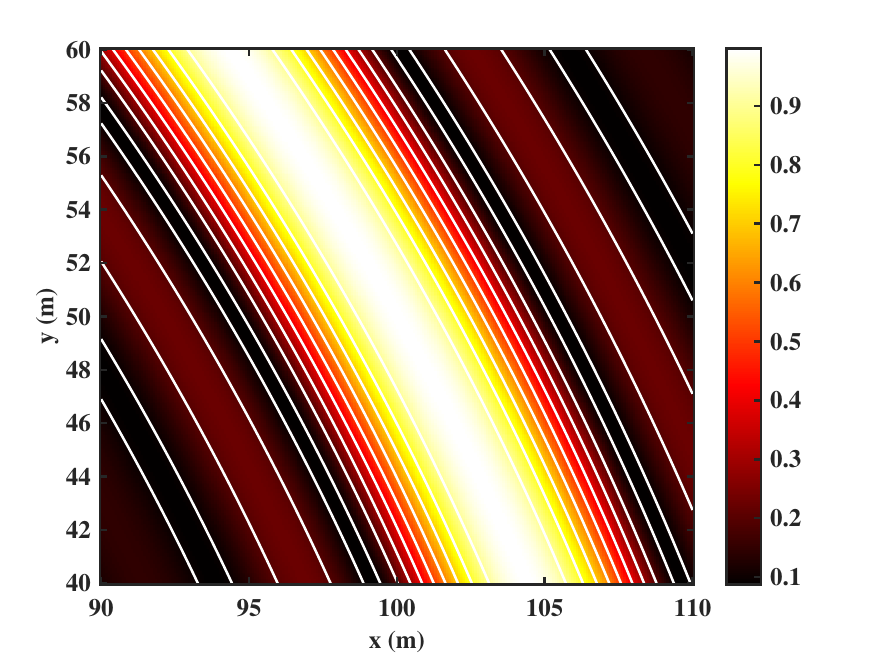}}
	\subfigure[$L = 1$] {\label{fig10.b}\includegraphics[width=.22\textwidth]{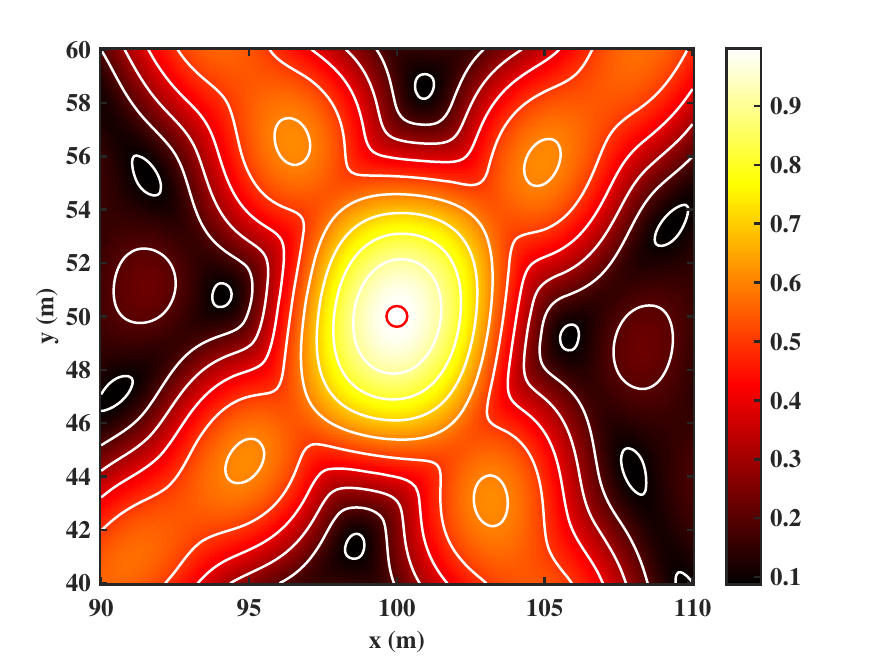}}
    \subfigure[$L = 2$] {\label{fig10.c}\includegraphics[width=.22\textwidth]{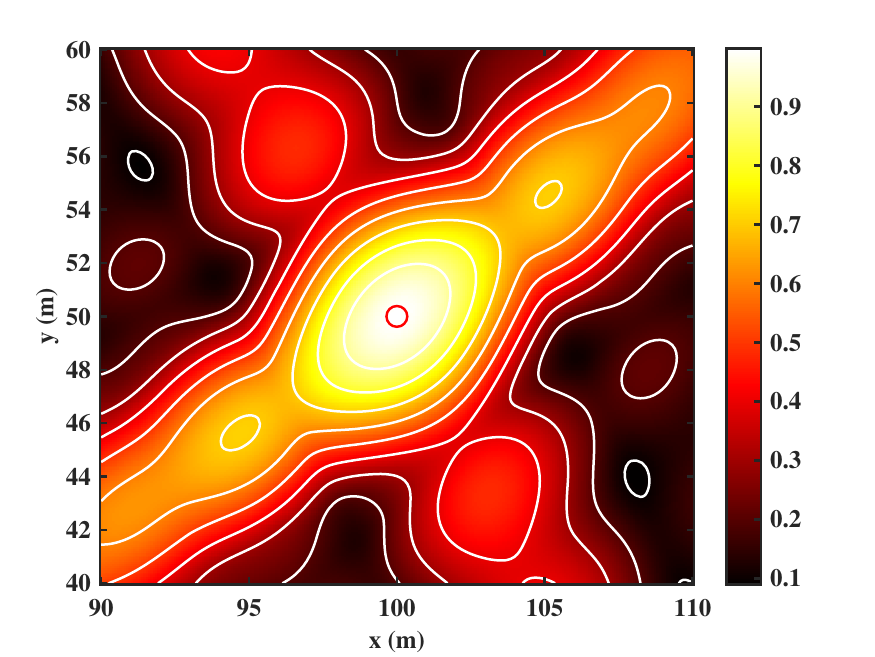}}
    \subfigure[$L = 3$] {\label{fig10.d}\includegraphics[width=.22\textwidth]{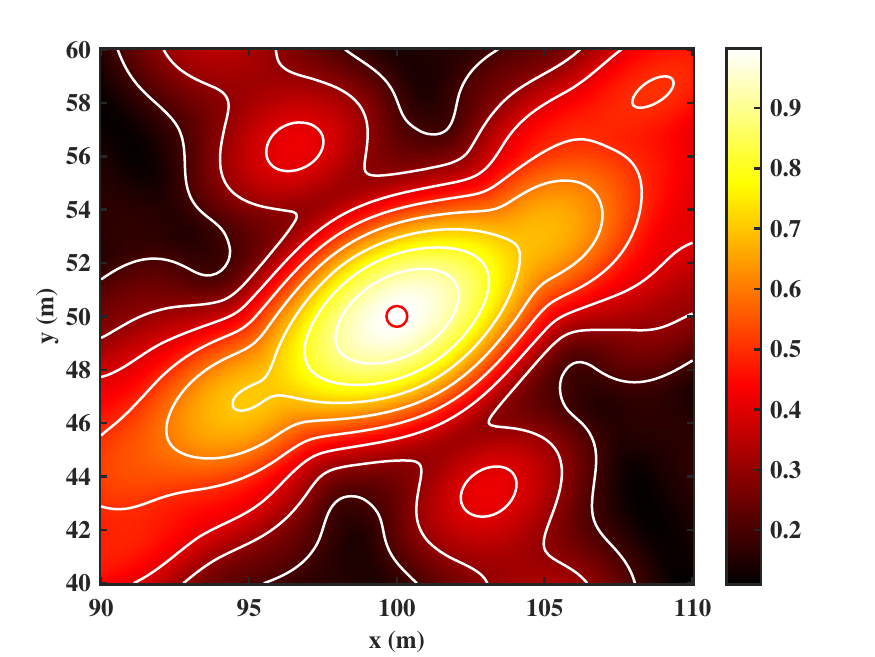}}
	\caption{localization results with different MPC numbers in SNR = -5 dB}
	\label{fig10}
\end{figure}

\begin{figure}[htbp]
	\centering
	\subfigure[$L = 0$] {\label{fig11.a}\includegraphics[width=.24\textwidth]{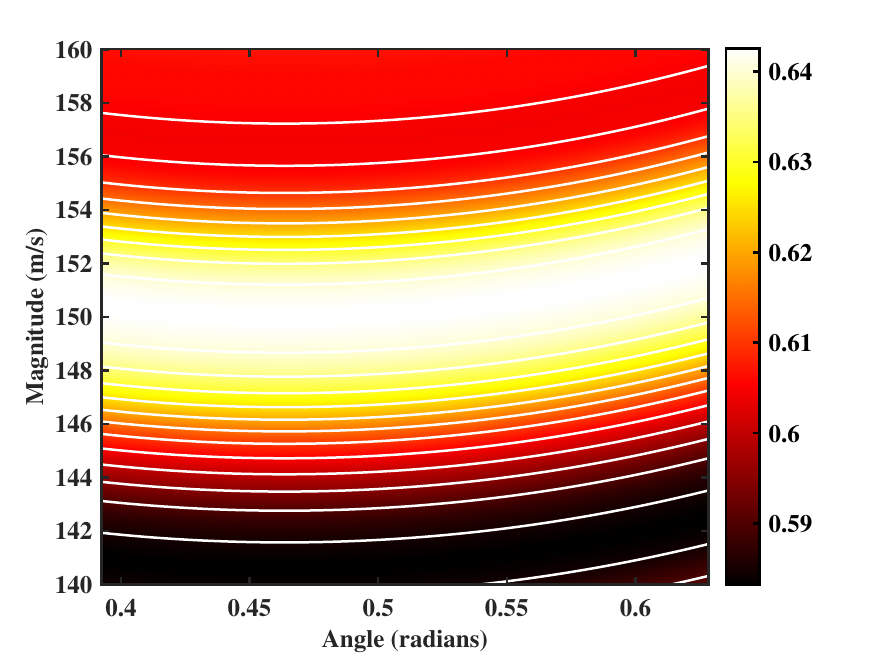}}
	\subfigure[$L = 1$] {\label{fig11.b}\includegraphics[width=.24\textwidth]{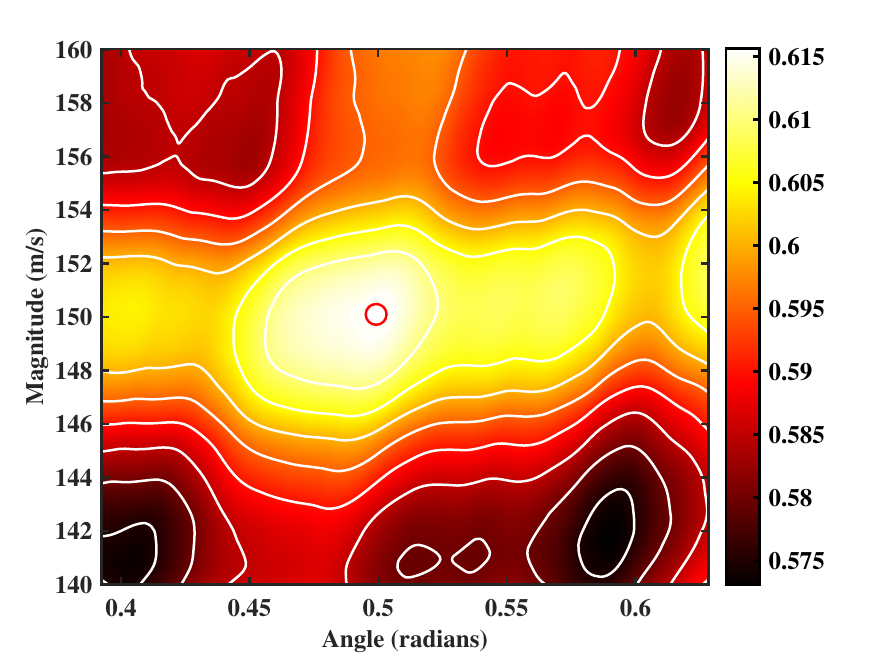}}
    \subfigure[$L = 2$] {\label{fig11.c}\includegraphics[width=.24\textwidth]{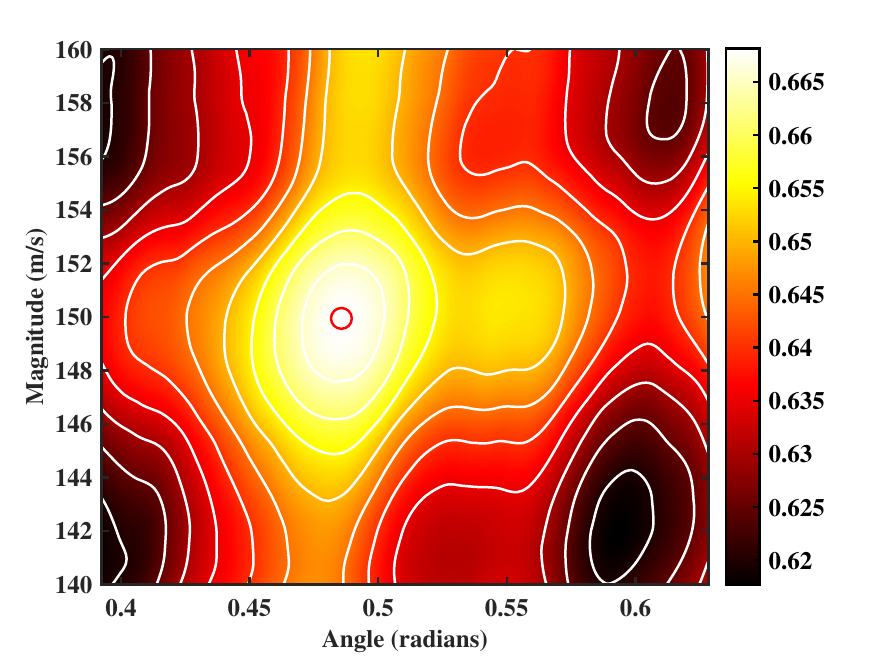}}
    \subfigure[$L = 3$] {\label{fig11.d}\includegraphics[width=.24\textwidth]{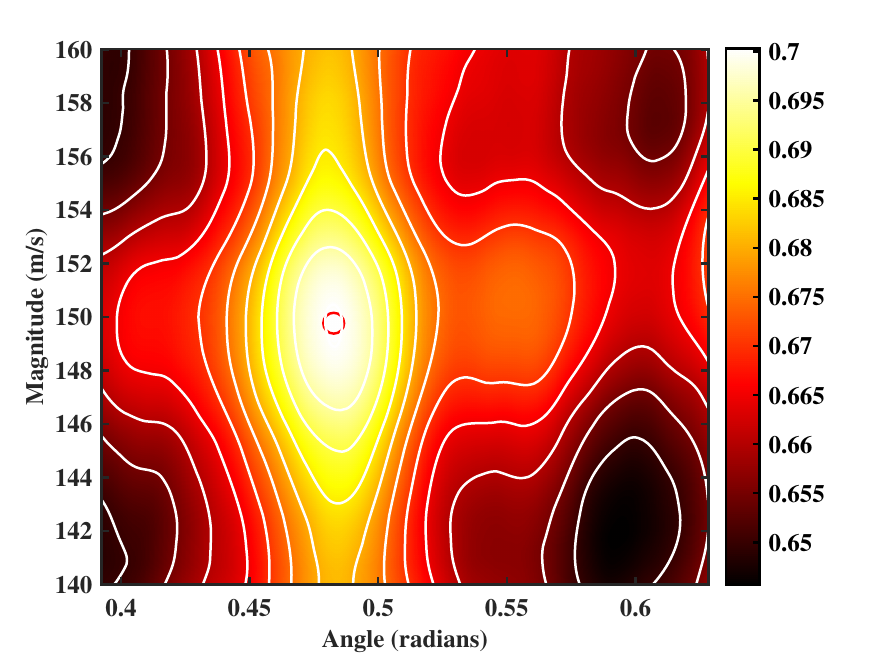}}
	\caption{Absolute velocity estimation results with different MPC numbers in SNR = -5 dB}
	\label{fig11}
\end{figure}

\subsubsection{Comparison of ARMSE}
Fig.~\ref{fig12} demonstrates the ARMSEs of location and absolute velocity estimations for the proposed SL-MPS scheme, 
with the existing scheme in \cite{gong2022multipath} as a benchmark.
To ensure fairness, 
the schemes differ only in signal processing after matching, 
emphasizing the differences between symbol-level and data-level fusion.

Upon observing Fig.~\ref{fig12}, 
the centimeter-level localization and absolute velocity estimation with errors below 0.1 m/s are achieved when the SNR exceeds 9 dB, 
demonstrating the feasibility, robustness, and high accuracy of the proposed SL-MPS scheme.

Besides, for both location and absolute velocity estimations, 
the proposed scheme achieves lower ARMSEs compared to the existing scheme, 
due to the SNR gain from symbol-level fusion, 
demonstrating the superiority of the proposed SL-MPS scheme in both location and absolute velocity estimation.
\begin{figure}[htbp]
	\centering
	\subfigure[Location estimation] {\label{fig12.a}\includegraphics[width=.30\textwidth]{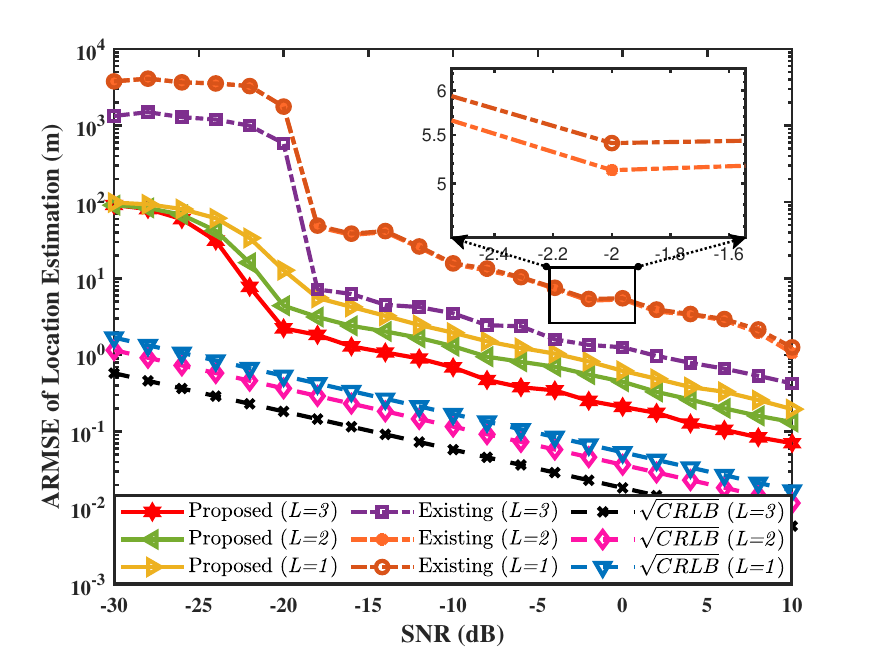}}
	\subfigure[Velocity estimation] {\label{fig12.b}\includegraphics[width=.30\textwidth]{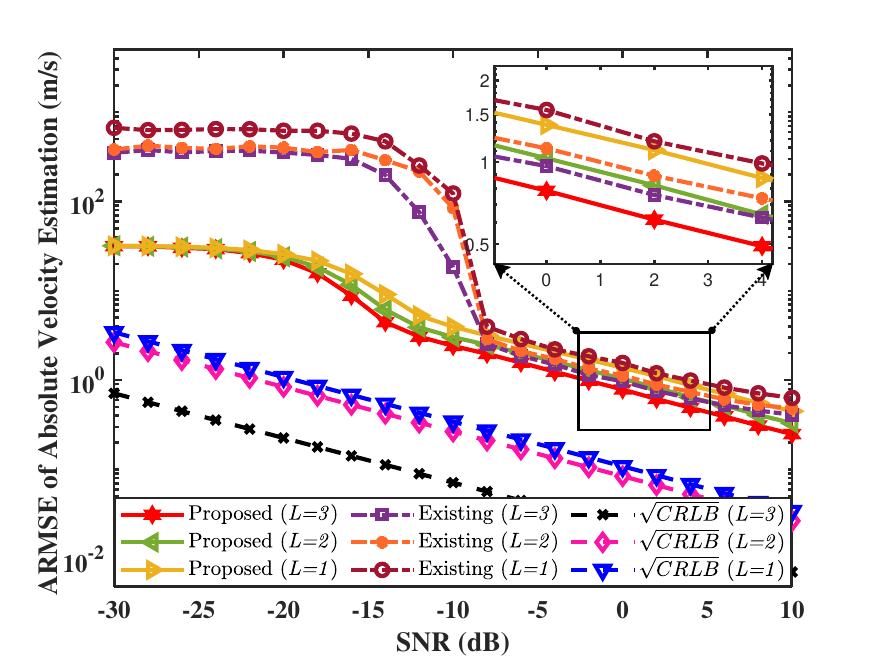}}
	\caption{Comparison of ARMSEs for location and velocity estimations}
	\label{fig12}
\end{figure}

\subsection{KRST Code and Dual-Function Performance Analysis}
For communication, 
we simulate different cases in Fig. \ref{fig13} to investigate the impact of the KRST code and MPCs on bit error rate (BER). 
The transmission rates under different cases can be summarized as follows:
\begin{itemize}
    \item Cases 6 and 8 achieve the highest transmission rate, with a value of 8 (bit/code length).
    \item Cases 1, 3, 5, and 7 exhibit an intermediate transmission rate, with a value of 4 (bit/code length).
    \item Cases 2 and 4 yield the lowest transmission rate, with a value of 2 (bit/code length).
\end{itemize}

Simulation results in Fig.~\ref{fig13} show that increasing MPCs (e.g., in case 1 and case 3) or the modulation order (e.g., in case 4 and case 7) raises the BER, 
reducing communication reliability. 
Conversely, increasing the code length (e.g., in case 1 and case 2) lowers the BER and improves reliability. 
For the same transmission rate (e.g., in case 1 and case 5), 
different combinations of modulation order and code length lead to varying performance. 
Thus, selecting the appropriate combination of modulation order and code length is crucial for optimizing communication performance.
\begin{figure}
    \centering
\includegraphics[width=0.30\textwidth]{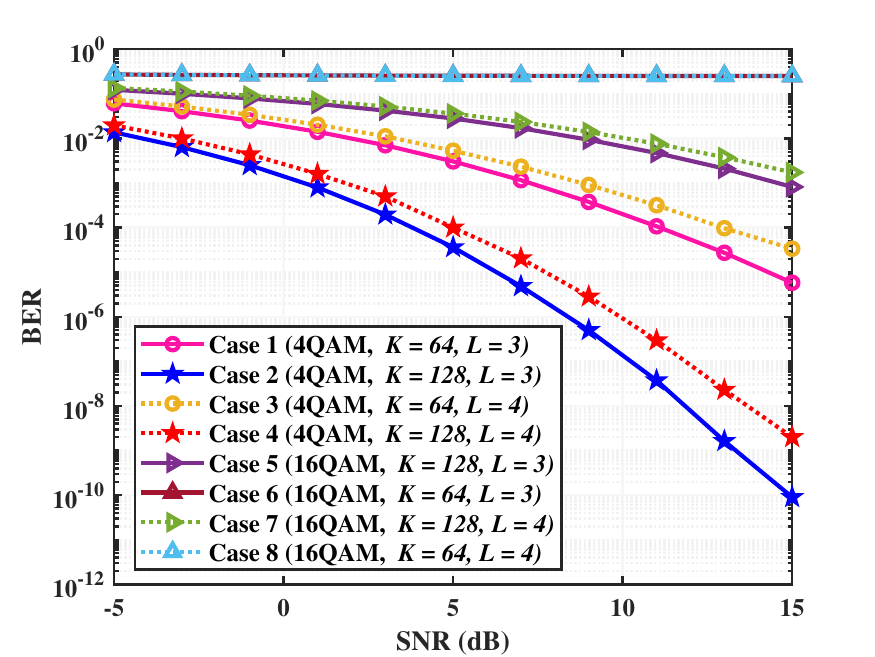}
    \caption{BER for different cases}
    \label{fig13}
\end{figure}

For sensing, 
we simulate the ARMSEs for location and absolute velocity estimations with various code lengths $K$ in Fig.~\ref{fig14}. 
The ARMSEs of both location and absolute velocity estimations decreased as $L$ increased, 
confirming the contribution of the KRST code to sensing performance.

Furthermore, we observe an interesting phenomenon in Fig. \ref{fig14}, 
where the four curves approximately follow an inverse square root function. 
This is because the CRLB is inversely proportional to the SNR, 
implying that the code gain $K$ is related to the ARMSE by an inverse square root relationship. 
Hence, when selecting $K$, 
focusing on the region with a steeper gradient allows smaller increases in $K$ to achieve a significantly lower ARMSE.
\begin{figure}
    \centering
\includegraphics[width=0.30\textwidth]{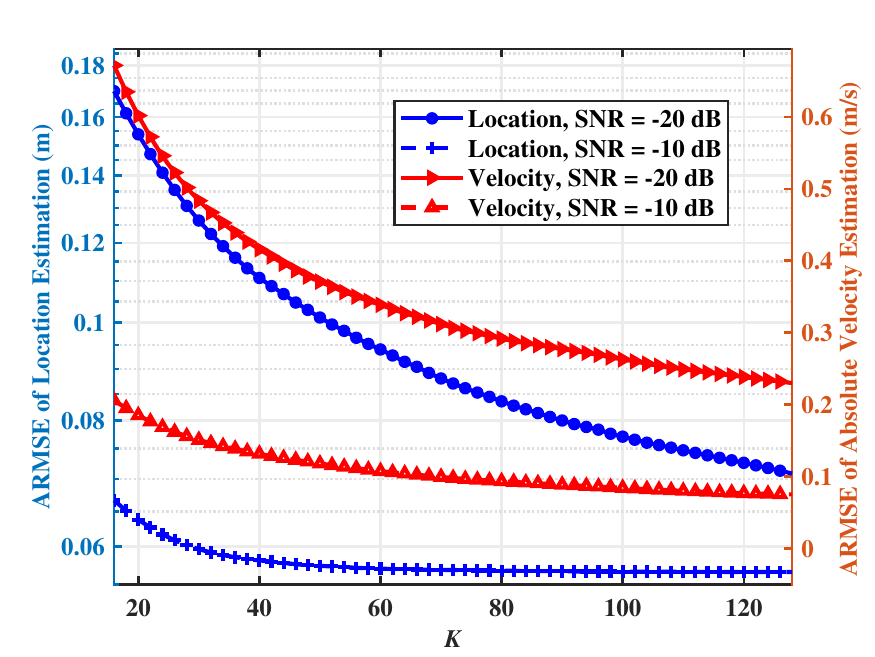}
    \caption{ARMSEs for location and absolute velocity estimations with various $K$}
    \label{fig14}
\end{figure}

Finally, we explore the trade-off between communication and sensing performance for various code lengths $K$ in Fig.~\ref{fig15}. 
Sensing performance is defined as the inverse of the weighted sum of the CRLBs for location and absolute velocity estimations, 
while transmission rate is $\frac{\left(\frac{N_\text{T}}{K}\right)\log_2\left(\varphi \right)}{T\Delta f}$ bps/Hz~\cite{3gpp2018nr}.

According to Fig.~\ref{fig15}, as $K$ increases, 
sensing performance increases while communication performance decreases. 
Therefore, the selection of $K$ can determine a suggested blur interval based on the trend in Fig.~\ref{fig15} 
and the observed phenomenon in Fig.~\ref{fig14}.
\begin{figure}
    \centering
\includegraphics[width=0.30\textwidth]{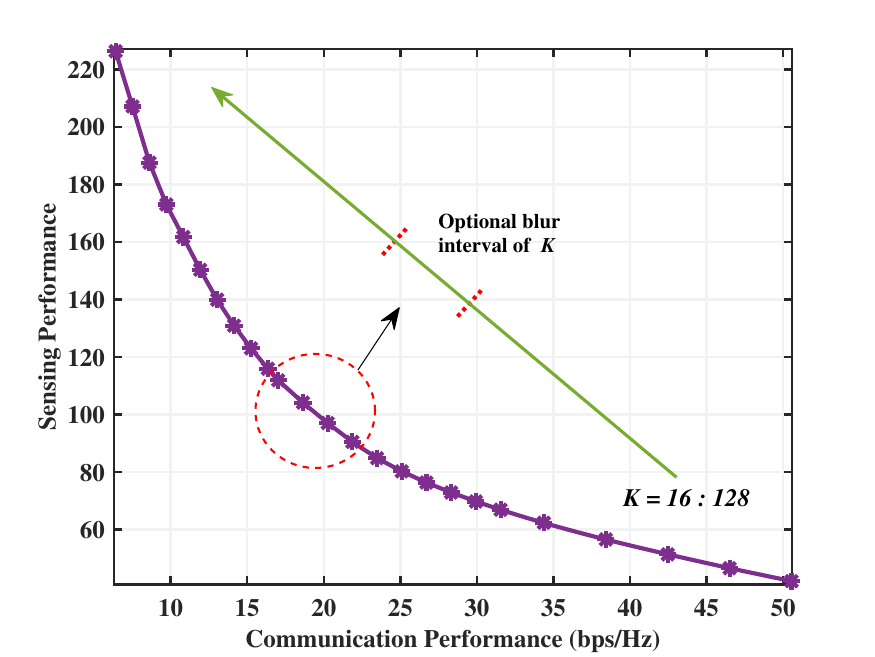}
    \caption{Trade-off performance between communication and sensing with various code length $K$}
    \label{fig15}
\end{figure}

\section{Conclusion}\label{se6}
In this paper, we propose the leverage of MPC information and KRST code in single MIMO-OFDM ISAC BS to realize high-accuracy multiple dynamic targets sensing and multi-user communication. 
Specifically, KRST code is leveraged to mitigate MPC interference in communication while enhancing sensing performance. 
Additionally, we present a novel SL-MPS scheme that significantly improves the localization capability of a single ISAC system and 
empowers a single BS with the new capability of absolute velocity estimation of targets with a single sensing attempt. 
In the SL-MPS scheme, ISAC echo signal is constructed in the form of third-order tensor, 
and a TP-ALS algorithm is proposed to realize MPC separation. 
Through VAs and the matching operation, 
the multiple dynamic target sensing problem is transformed into single-target sensing problems, 
which are solved via proposed localization and absolute velocity estimation methods. 
Furthermore, closed-form CRLBs for localization and absolute velocity with MPC information are derived, 
revealing the contributions of MPC information to dynamic-target sensing. 
Simulation results validate the feasibility of the proposed SL-MPS scheme and demonstrate its superior sensing performance 
compared to state-of-the-art MPC-enhanced schemes.

\begin{appendices}
\section{Exhaustive search algorithm} \label{apex1}
For convenience, 
we assume that BS is the $0$-th VA. 
For the $l$-th VA, the received shuffled range information can be expressed as
\begin{equation}\label{ap1}
\mathbf{R}_{l}=\mathbf{d}_{l}\mathbf{P}_{l}\in\mathbb{R}^{1\times I},
\end{equation}
where $\mathbf{P}_{l}\in\mathbb{R}^{I\times I}$ is a permutation matrix, 
$\mathbf{d}_{l}=\left[d_{l}^1,d_{l}^2,\cdots,d_{l}^I\right]$ is the range vector that satisfies the one-to-one relationship, 
and the $i$-th element of $\mathbf{R}_{l}$ is $\hat{\tau}_{i(l+1)}c_0$. 
The goal of exhaustive search algorithm is to find $\{\mathbf{P}_{l}\}_{l=1}^{L}$, 
and the detailed operations are as follows.

1) Grid the range of interest to obtain $\varpi \times \varpi$ possible target positions, 
    where the $\{q, s\}$ position is represented as $(x_q, y_s)$ with $q, s\in\mathbb{S}=\{1,2,\cdots,\varpi\}$.
    
2) Select $I$ from $\varpi \times \varpi$ possible target positions to form a position vector, 
    where the $\varsigma\in\{1,2,\cdots,\frac{(\varpi \times \varpi)!}{I!(\varpi \times \varpi-I)!}\}$-th position vector is $\mathbf{S}_\varsigma=\left[(x_{a_1}, y_{b_1}), (x_{a_2}, y_{b_2}),\cdots, (x_{a_I}, y_{b_I})\right]$,
    with $a_i, b_i=\text{randint}(1,\varpi)$, $\text{randint}(1,\varpi)$ representing a uniformly distributed random integer generated in the $\left[1, \varpi\right]$, and $(\cdot)!$ denoting the factorial operator.
    
3) For each VA, calculate the distance vector and the $l$-th distance vector is $\mathbf{D}_{l}=\left[d_1^{l},d_2^{l},\cdots,d_I^{l}\right]$, where $d_i^{l}=\sqrt{(x_{a_i}-x_{\text{VA}}^{l})^2+(y_{b_i}-y_{\text{VA}}^{l})^2}$.

4) For each VA, a loss function $\mathbf{L}_{l}$ is denoted by
    \begin{equation}
        \mathbf{L}_{l}=\underset{\sigma\in S_I}{\min}\sum_{i=1}^I|d_{\sigma(i)}^{l}-\mathbf{R}_{l}(i)|,
    \end{equation}
    where $\sigma$ is a permutation of the indices $\{1,2,\cdots,I\}$ and $S_I$ is the set of all possible permutations.
    
5) The final exhaustive search problem is to find the optimal $\mathbf{S}_\varsigma$. The optimal permutation $\sigma_{l}$ for each VA that minimizes the followed total loss function is
    \begin{equation}
        \left(\hat{\mathbf{S}}, \hat{\sigma}_{0},\hat{\sigma}_{1},\cdots,\hat{\sigma}_{L}\right)= \underset{\left(\mathbf{S}, \sigma_{0},\sigma_{1},\cdots,\sigma_{L}\right)}{\text{argmin}}\sum_{l=0}^L\mathbf{L}_{l}.
    \end{equation}

With the $\{\hat{\sigma}_{0},\hat{\sigma}_{1},\cdots,\hat{\sigma}_{L}\}$, 
the one-to-one relationship between $I$ targets and $I$ ranges in each VA can be obtained.

\section{Absolute velocity estimation method of data-level fusion}\label{apex2}
we first obtain the Doppler information 
$\{\hat{f}_{\text{D},i(l+1)}\}_{l=0}^{L}$ 
by performing MUSIC method on the $(i,3)$-th element of matched MPC cluster matrix in each VA and the BS.
Then, the angles $\{\hat{\theta}_{i(l+1)}\}_{l=0}^{L}$ 
are obtained by performing MUSIC method on the $(i,1)$-th element of matched MPC cluster matrix in each VA and the BS. 
According to Fig.~\ref{fig4}, the angle from the 
$i$-th target to the ${l'}$-th VA can be expressed as
\begin{equation} 
\hat{\theta}_{i(l'+1),\text{VA}}=|\hat{\theta}_{i(l'+1)}|+2\phi_{l'}-2\pi,
\end{equation}
and the angle from the $i$-th target to the BS is $\hat{\theta}_{i(0+1),\text{VA}}=\hat{\theta}_{i(0+1)}$.

For the $i$-th moving target with magnitude $|\Vec{v}_i|$ and angle $\theta_\text{v}^i$, 
the estimated Doppler information in $\tilde{\Phi}_\text{VA}^{l'}$ for the $i$-th target can be expressed as
\begin{equation} \label{eq59}
\hat{f}_{\text{D},i(l'+1)}=\frac{-2f_\text{c}|\Vec{v}_i|}{c_0}\left[\cos(\hat{\theta}_{i(l'+1),\text{VA}}-\theta_\text{v}^i)\right],
\end{equation}
and the estimated Doppler information in $\tilde{\Phi}_\text{VA}^{0}$ for the $i$-th target 
is $\hat{f}_{\text{D},i(0+1)}=\frac{-2f_\text{c}|\Vec{v}_i|}{c_0}\left[\cos(\hat{\theta}_{i(0+1),\text{VA}}-\theta_\text{v}^i)\right]$.

The absolute velocity estimation problem of the $i$-th target is transformed into an MLE problem, 
and the estimated error vector is expressed as
\begin{equation}
   {\fontsize{8}{8} \mathbf{v}_\text{ta}=\left[
        \begin{array}{c}
          \hat{f}_{\text{D},i(0+1)}+\frac{2f_\text{c}|\Hat{\Vec{v}}_i^a|}{c_0}\left[\cos(\hat{\theta}_{i(0+1),\text{VA}}-\hat{\theta}_\text{a}^i)\right]    \\
          \hat{f}_{\text{D},i(1+1)}+\frac{2f_\text{c}|\Hat{\Vec{v}}_i^a|}{c_0}\left[\cos(\hat{\theta}_{i(1+1),\text{VA}}-\hat{\theta}_\text{a}^i)\right]     \\
          \vdots \\
          \hat{f}_{\text{D},i(L+1)}+\frac{2f_\text{c}|\Hat{\Vec{v}}_i^a|}{c_0}\left[\cos(\hat{\theta}_{i(L+1),\text{VA}}-\hat{\theta}_\text{a}^i)\right]              
        \end{array}
    \right],}
\end{equation}
and the convex optimization problem is
\begin{equation}
\left[|\Hat{\Vec{v}}_i|,\hat{\theta}_\text{v}^i\right]=\underset{\left[|\Hat{\Vec{v}}_i^a|,\hat{\theta}_\text{a}^i\right]}{\text{argmin}}\sum_{\zeta =1}^{L+1}\|\text{v}_\text{ta}(\zeta )\|_2,
\end{equation}
which is solved by the Broyden-Fletcher-Goldfarb-Shanno (BFGS) algorithm~\cite{liu2023joint}, 
and the estimated absolute velocity of the $i$-th target is $(\hat{\theta}_\text{v}^i,|\Hat{\Vec{v}}_i|)$.

\end{appendices}

\bibliographystyle{IEEEtran}
\bibliography{main}

\end{document}